\pdfoutput=1
\def\VersionLong{}
\def\VersionFinal{}

\newif\ifsupresstikzfig
\supresstikzfigfalse

\ifdefined\LaTeXdiff
	\let\VersionWithComments\undefined
	\let\WithReply\undefined
\fi

\ifdefined\VersionLong%
	\newcommand{\LongVersion}[1]{#1}
	\newcommand{\ShortVersion}[1]{}
\else
	\newcommand{\LongVersion}[1]{}
	\newcommand{\ShortVersion}[1]{#1}
\fi

\PassOptionsToPackage{svgnames,table}{xcolor}
\documentclass[manuscript,acmsmall,\ifdefined\VersionFinal%
\else%
review,anonymous\fi,screen\LongVersion{,authorversion,nonacm}]{acmart}

\usepackage{comment}
\ifdefined\VersionLong%
    \includecomment{LongVersionBlock}
    \excludecomment{ShortVersionBlock}
\else
    \excludecomment{LongVersionBlock}
    \includecomment{ShortVersionBlock}
\fi
\ifdefined\VersionFinal%
    \includecomment{FinalVersionBlock}
\else
    \excludecomment{FinalVersionBlock}
\fi

\usepackage[utf8]{inputenc}

\AtBeginDocument{%
  }

\ShortVersion{\setcopyright{acmlicensed}\acmJournal{TECS}
\acmYear{2023} \acmVolume{1} \acmNumber{1} \acmArticle{1} \acmMonth{1} \acmPrice{15.00}\acmDOI{10.1145/3609127}}
\LongVersion{\setcopyright{none}}

\usepackage[ruled,lined,linesnumbered,noend]{algorithm2e}
\SetKwInOut{Input}{input}
\SetKwInOut{Output}{output}
\SetKw{KwRemove}{remove}
\SetKw{KwAdd}{add}
\SetKw{KwPush}{push}
\SetKw{KwPop}{pop}
\SetKw{KwLet}{let}
\SetKw{KwBe}{be}
\SetKw{KwFrom}{from}
\SetKw{KwTo}{to}
\SetKw{KwCompute}{compute}
\SetKw{KwReturn}{return}
\SetKw{KwBreak}{break}
\SetKw{KwPick}{pick}
\SetKw{KwCase}{case}
\SetKwProg{Fn}{Function}{:}{}
\SetAlgorithmName{Algorithm}{algorithm}{List of Algorithms}

\usepackage{subcaption}
\usepackage{xspace}

\usepackage{stmaryrd}
\usepackage{multirow}

\graphicspath{{figures/}}

\ifdefined\VersionWithComments
	\usepackage{draftwatermark}
	\SetWatermarkText{draft}
	\SetWatermarkScale{3}
	\SetWatermarkColor[gray]{0.9}
\fi

\usepackage[capitalise,english,nameinlink]{cleveref} 
\crefname{line}{\text{line}}{\text{lines}} 
\Crefname{line}{\text{Line}}{\text{Lines}} 
\crefname{item}{\text{item}}{\text{items}} 
\crefname{example}{\text{Example}}{\text{Examples}} 
\crefname{assumption}{\text{Assumption}}{\text{Assumptions}} 
\crefname{algorithm}{\text{Algorithm}}{\text{Algorithms}}

\usepackage{wrapfig}
\usepackage{tikz}
\usetikzlibrary{arrows,arrows.meta,automata,positioning,calc,math}
\tikzstyle{every node}=[initial text=]
\tikzstyle{location}=[rectangle, rounded corners, minimum size=12pt, draw=black, fill=blue!10, inner sep=2pt]
\tikzstyle{final}=[double]
\tikzstyle{accepting}=[final]

\tikzstyle{defproblem} = [
 draw=black,
 fill=cyan!10,
 text=black,
 line width=0.5pt,
  text width = \linewidth - 1.6 ex - 1pt,
  inner sep = 0.8 ex,
  rounded corners=4pt]
\newcommand{\defProblem}[3]
{%
\smallskip
\noindent\begin{tikzpicture}%
\draw node[defproblem]{%
\textbf{#1 problem:}\\%
\textsc{Input}: #2\\%
\textsc{Problem}: #3};%
\end{tikzpicture}}%
\ifdefined\VersionWithComments%
	\usepackage[colorinlistoftodos,textsize=tiny]{todonotes}
\else
	\usepackage[disable]{todonotes}
\fi
\newcommand{\gennote}[3]{\todo[linecolor=#2,backgroundcolor=#2!25,bordercolor=#2]{#3: #1}}
\newcommand{\js}[1]{\gennote{#1}{blue}{JS}}
\newcommand{\mw}[1]{\gennote{#1}{orange}{MW}}
\newcommand{\ks}[1]{{\gennote{#1}{purple}{KS}}}
\newcommand{\instructions}[1]{{\gennote{\bfseries #1}{red}{Instructions}}}
\newcommand{\reviewer}[2]{{\gennote{``#2''}{purple}{Reviewer #1}}}

\ifdefined\VersionWithComments%
 	\definecolor{colorok}{RGB}{80,80,150}
\else
	\definecolor{colorok}{RGB}{0,0,0}
\fi

\newcommand{\eg}{\textcolor{colorok}{e.\,g.,}\xspace}
\newcommand{\ie}{\textcolor{colorok}{i.\,e.,}\xspace}

\definecolor{coloract}{named}{black}
\definecolor{colorclock}{named}{black}
\definecolor{colorconst}{named}{black}
\definecolor{colordisc}{named}{black}
\definecolor{colorloc}{rgb}{0.4, 0.4, 0.65}
\definecolor{colorparam}{named}{black}

\newif\iftikzgnuplot
\tikzgnuplottrue
\usepackage{gnuplot-lua-tikz}

\newcommand{\powerset}[1]{\mathcal{P} ({#1})}
\newcommand{\subseteqfin}{\mathrel{\subseteq_{\mathrm{fin}}}}
\newcommand{\setdiff}{\triangle}

\newcommand{\N}{\mathbb{N}}
\newcommand{\AP}{\mathbf{AP}}
\newcommand{\Next}{\mathcal{X}}
\newcommand{\Until}[1][]{\mathrel{\mathcal{U}_{#1}}}
\newcommand{\Glb}{\Box}
\newcommand{\Evt}{\Diamond}
\newcommand{\Pmax}{\mathrel{\mathbb{P}_{\mathrm{max=?}}}}

\newcommand{\Mealy}{\hat{\mathcal{M}}}
\newcommand{\SUT}{\mathcal{M}}
\newcommand{\MDP}{\hat{\mathcal{M}}}
\newcommand{\MDPi}[1]{\hat{\mathcal{M}}^{#1}}
\newcommand{\compose}[2]{{#1}_{#2}}
\newcommand{\Lg}{\mathcal{L}}
\newcommand{\targetLg}{\mathcal{L}_{\mathrm{tgt}}}
\newcommand{\Alphabet}{\Sigma}
\newcommand{\A}{\mathcal{A}}
\newcommand{\hypothesisA}{\mathcal{A}_{\mathrm{cnd}}}
\newcommand{\emptyword}{\varepsilon}
\newcommand{\word}{w}
\newcommand{\cex}{\mathit{cex}}

\newcommand{\Spec}{\varphi}
\newcommand{\strategy}{s}
\newcommand{\trace}{\sigma}
\newcommand{\traceWithoutSuffix}{\sigma^{-}}

\newcommand{\otraceSingle}{o}

\newcommand{\INPUT}{\Sigma^{\mathrm{in}}}
\newcommand{\OUTPUT}{\Sigma^{\mathrm{out}}}
\newcommand{\action}{a}
\newcommand{\observation}{b}
\newcommand{\Dist}[1][]{\mathrel{\mathit{Dist}({#1})}}
\newcommand{\Prob}{\mathbb{P}}

\newcommand{\prefix}{p}
\newcommand{\PrefixSet}{P}
\newcommand{\ExtPrefixSet}{(\PrefixSet \cdot \Alphabet)}
\newcommand{\suffix}{s}
\newcommand{\SuffixSet}{S}

\newcommand{\TracePool}{\mathcal{S}}
\newcommand{\COMPATIBLE}{\mathbf{compatible}}
\newcommand{\EQROW}{\mathbf{eqRow}}

\newcommand{\ourTool}{ProbBBC}
\newcommand{\ourToolOnlyClassic}{ProbBBC w/o str.\ comp.}
\newcommand{\baselineMethod}{\textsf{ProbBlackReach}}
\newcommand{\Lstar}{\ensuremath{\mathrm{L^{*}}}}
\newcommand{\LstarMDP}{\ensuremath{\Lstar_{\mathrm{MDP}}}}
\newcommand{\Aalergia}{\textsc{Aalergia}}

\newcommand{\SlotMachine}{\textsf{Slot}}
\newcommand{\SlotMachineWithSuppressedOutputs}{\textsf{SlotLimitedObs}}
\newcommand{\MQTT}{\textsf{MQTT}}
\newcommand{\TCP}{\textsf{TCP}}
\newcommand{\FirstGridWorld}{\textsf{GridWorldSmall}}
\newcommand{\SecondGridWorld}{\textsf{GridWorldLarge}}
\newcommand{\SharedCoin}{\textsf{SharedCoin}}
\newcommand{\RandomGridWorld}{\textsf{RandomGridWorld}}

\newcommand\FIN{\mathit{fin}}

\usepackage{colortbl}

\newcommand{\goodCell}{\cellcolor{green!25}\bf}

\tikzstyle{rqanswer} = [
 draw=black,
 fill=gray!30,
 text=black,
 line width=0.5pt,
  text width = \linewidth - 1.6 ex - 1pt,
  inner sep = 0.8 ex,
  rounded corners=4pt]
\newcommand{\rqanswer}[2]{\LongVersion{\smallskip}\noindent%
\begin{tikzpicture}%
\draw node[rqanswer]{\textbf{Answer to {#1}}:{ #2}};%
\end{tikzpicture}\vspace{0em}}

\begin{document}

\title{Probabilistic Black-Box Checking via Active MDP Learning}
\thanks{\ifdefined\VersionLong%
This is the author version of the paper of the same name accepted to International Conference on Embedded Software (EMSOFT), 2023%
\else%
This article appears as part of the ESWEEK-TECS special issue and was presented in the International Conference on Embedded Software (EMSOFT), 2023\fi}

\author{Junya Shijubo}
\email{shijubo@fos.kuis.kyoto-u.ac.jp}
\orcid{0000-0002-2853-1159}
\author{Masaki Waga}
\email{mwaga@fos.kuis.kyoto-u.ac.jp}
\orcid{0000-0001-9360-7490}
\author{Kohei Suenaga}
\orcid{0000-0002-7466-8789}
\email{ksuenaga@gmail.com}

\affiliation{%
  \institution{Kyoto University}
  \streetaddress{Yoshida-honmachi, Sakyo-ku}
  \city{Kyoto}
  \postcode{606-8501}
  \country{Japan}
}

\renewcommand{\shortauthors}{Junya Shijubo, Masaki Waga, and Kohei Suenaga}

\begin{abstract}  
We introduce a novel methodology for testing stochastic black-box systems, frequently encountered in embedded systems. Our approach enhances the established \emph{black-box checking (BBC)} technique to address stochastic behavior. Traditional BBC primarily involves iteratively identifying an input that breaches the system's specifications by executing the following three phases: the \emph{learning} phase to construct an automaton approximating the black box's behavior, the \emph{synthesis} phase to identify a candidate counterexample from the learned automaton, and the \emph{validation} phase to validate the obtained candidate counterexample and the learned automaton against the original black-box system. Our method, ProbBBC, refines the conventional BBC approach by (1) employing an active Markov Decision Process (MDP) learning method during the learning phase, (2) incorporating probabilistic model checking in the synthesis phase, and (3) applying statistical hypothesis testing in the validation phase. ProbBBC uniquely integrates these techniques rather than merely substituting each method in the traditional BBC; for instance, the statistical hypothesis testing and the MDP learning procedure exchange information regarding the black-box system's observation with one another. The experiment results suggest that ProbBBC outperforms an existing method, especially for systems with limited observation.
\end{abstract}



\received{20 February 2007}
\received[revised]{12 March 2009}
\received[accepted]{5 June 2009}

\maketitle

\ifdefined\VersionWithComments%
	\textcolor{red}{\textbf{This is the version with comments. To disable comments, comment out line~3 in the \LaTeX{} source.}}
\fi

\instructions{(EMSOFT 2023) A blind review process will be enforced. Authors must not reveal their identity directly or indirectly. Papers should not exceed 20 pages in ACM Transactions one-column format for the initial submission. References, appendix, etc. are all counted towards the 20-page limit. If the paper is recommended for revision after the first round of reviews, then the revised manuscript should not exceed 25 pages in the same format.}

\js{hello}
\mw{hello}
\ks{hello}

\section{Introduction}\label{section:introduction}

Embedded systems (ESs), including cyber-physical systems (CPSs) and IoT systems, are often safety-critical systems, making their safety assurance crucial.
One of the characteristics of such systems is that their behavior is governed by stochastic disturbances (\eg{} due to communication error or uncertainty in physical environments), not only by the external inputs (\eg{} from a controller or other agents).
Another notable characteristic is that they are often black-box systems due to their proprietary or machine-learned components.  These characteristics make testing and verification of such systems challenging.

\emph{Black-box checking (BBC)}~\cite{DBLP:conf/forte/PeledVY99} is a technique to apply model checking to black-box systems via model learning. 
Given a black-box system $\SUT$ under test (SUT) and a temporal logic formula $\Spec$, it tries to find an input sequence witnessing $\SUT \not\models \Spec$ by iterating the following three phases:
\begin{description}
\item[Learning phase] learns an automaton $\Mealy$ that approximates the behavior of $\SUT$ by using an active learning procedure such as \Lstar{}~\cite{DBLP:journals/iandc/Angluin87}.
\item[Synthesis phase] tries to synthesize an execution trace $\trace$ \LongVersion{witnessing}\ShortVersion{showing} $\Mealy \not\models \Spec$ by model checking.
\item[Validation phase] validates whether the found $\trace$ also witnesses $\SUT \not\models \Spec$.  If $\trace$ does not witness $\SUT \not\models \Spec$, $\trace$ is an evidence of $\SUT \neq \Mealy$, and BBC goes back to the learning phase using $\trace$ to refine $\Mealy$. When no witness for $\Mealy \not\models \Spec$ is found in the synthesis phase, this phase tries to find an evidence of $\SUT \neq \Mealy$, typically through random testing.
\end{description}

There have been various extensions of BBC~\cite{DBLP:conf/tacas/GrocePY02,DBLP:conf/forte/ElkindGPQ06,DBLP:journals/isse/MeijerP19}, including the ones focusing on CPSs~\cite{DBLP:conf/hybrid/Waga20,DBLP:conf/rv/ShijuboWS21}.
However, these techniques are limited to \emph{deterministic} systems because the learned automaton is \emph{deterministic}, \eg{} a DFA or a Mealy machine.

\emph{Markov decision processes (MDP)}, an extension of Markov chains so that external inputs control the transition distribution, are widely used to formalize systems with external inputs and probabilistic branching, 
By fixing a \emph{strategy} $\strategy$, a function to decide the input to be fed to the system, a Markov chain $\compose{\MDP}{\strategy}$ can be constructed from an MDP $\MDP$.

Given an MDP $\MDP$ modeling a system and a temporal logic formula $\Spec$ representing the specification, \emph{(quantitative) probabilistic model checking (PMC)}~\cite{DBLP:reference/mc/BaierAFK18} computes the maximum (or minimum) satisfaction probability $p_{\compose{\MDP}{\strategy},\Spec}$ of $\Spec$ on $\MDP$ as well as a strategy $\strategy$ to realize it. Such probability can be used as a safety measure of the system under verification. 
Although PMC returns the exact probability, it is challenging to apply it to real-world ESs because they are rarely white-box, and modeling the system under verification as an MDP is a non-trivial task.

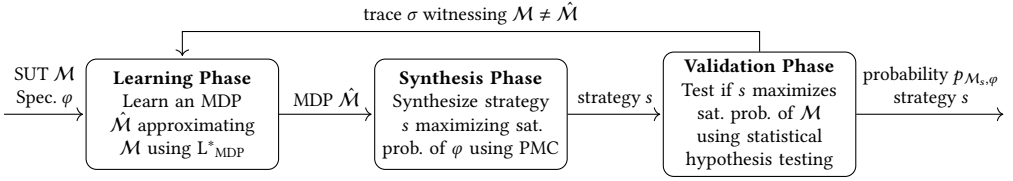
\begin{figure}[tbp]
 \centering
 \begin{tikzpicture}[shorten >=1pt,align=center,node distance=1.75cm,scale=0.725,every node/.style={transform shape}]
  \tikzset{phase/.style={
    rectangle,
    draw, rounded corners,
    text width=3.3cm,
    minimum width=3cm,
    minimum height=2cm
  }}
  \node (start) at (0,0) {};
  \node[phase,node distance=1.5cm,right=of start] (learning) {\textbf{Learning Phase}\\Learn an MDP $\MDP$ approximating $\SUT$ using \LstarMDP{}};
  \node[phase,right=of learning] (synthesis) {\textbf{Synthesis Phase}\\Synthesize strategy $\strategy$ maximizing sat.\ prob.\ of $\Spec$ using PMC};
  \node[phase,right=of synthesis] (validation) {\textbf{Validation Phase}\\Test if $\strategy$ maximizes sat. prob. of $\SUT$ using statistical hypothesis testing};
  \node[node distance=2.75cm,right=of validation] (end) {};

  \node[node distance=1cm,above=of validation] (validation_above) {};
  \node[node distance=1cm,above=of learning] (learning_above) {};

  \path[->]
  (start) edge node[above,align=center]{SUT $\SUT$\\Spec. $\Spec$} (learning)
  (learning) edge node[above] {MDP $\MDP$} (synthesis)
  (synthesis) edge node[above] {strategy $\strategy$} (validation)
  (validation) edge node[above,align=center] {probability $p_{\compose{\SUT}{\strategy},\Spec}$\\strategy $\strategy$} (end)
  ;
  \draw[->] (validation) -- ($(validation) + (0,1.5)$) node[above] at ($(synthesis) + (0,1.5)$) {trace $\trace$ witnessing $\SUT \neq \MDP$} -- ($(learning) + (0,1.5)$) -- (learning);
 \end{tikzpicture}
 \caption{Outline of probabilistic black-box checking. Given a black-box SUT $\SUT$ and an LTL formula $\Spec$, it estimates the maximum satisfaction probability $p_{\compose{\SUT}{\strategy},\Spec}$ of $\Spec$ on $\SUT$ as well as a strategy $\strategy$ to realize it.}%
 \label{figure:illustration}
\end{figure}

In this paper, we propose \emph{probabilistic black-box checking (ProbBBC)}, a quantitative extension of BBC for stochastic systems.
 Given a black-box SUT $\SUT$ and a \emph{linear temporal logic (LTL)}~\cite{DBLP:conf/banff/Vardi95} formula $\Spec$, ProbBBC returns an estimation of the maximum satisfaction probability $p_{\compose{\SUT}{\strategy},\Spec}$ of $\Spec$ with a strategy $\strategy$ to realize it.
ProbBBC can be used, for example, for the following purposes.

\begin{example}
 [ProbBBC for debugging]
 Consider an IoT system $\SUT$ that uses the MQTT protocol~\cite{MQTTv3} for communication. Due to the wireless network's instability, software bugs, or an unstable power supply, communication may not always be stable. For debugging such an issue, it is helpful to synthesize adversarial inputs leading to a communication error with a high probability. ProbBBC can synthesize such inputs by estimating the strategy that maximizes the probability of causing a communication error. If the probability returned by ProbBBC is small after a sufficiently long time, one can deem that the system $\SUT$ is working well by the convergence of ProbBBC.\@
\end{example}

\begin{example}
 [ProbBBC for controller synthesis]
 Consider a robot $\SUT$ moving on a field with various conditions, \eg{} concrete, grass, and mud. Our task is to create a controller for the robot to reach a specific goal. Due to sensing noise or slipping depending on the field's condition, the robot's movement may be probabilistic. Furthermore, it is highly challenging to formally model all the system and environmental details. ProbBBC can synthesize a near-optimal controller by estimating the strategy that maximizes the probability of reaching the goal within certain time steps. If it fails to synthesize a controller satisfying the specification after a sufficiently long time, it suggests that the task may be highly challenging by the convergence of ProbBBC.\@
\end{example}

\cref{figure:illustration} outlines ProbBBC, which, like conventional BBC, consists of the learning, synthesis, and validation phases.
In the learning phase, ProbBBC learns a deterministic MDP $\MDP$, \ie{} an MDP such that the current state is uniquely identified from a sequence of inputs and outputs, from the SUT $\SUT$ based on an active MDP learning algorithm \emph{\LstarMDP{}}~\cite{DBLP:journals/fac/TapplerA0EL21}.
\LstarMDP{} systematically feeds input sequences to the SUT $\SUT$, estimates the probabilistic distribution of the output observations, and identifies the discrete structure and transition probabilities of $\SUT$.
ProbBBC uses the resulting MDP $\MDP$ as an approximation of $\SUT$.\@

In the synthesis phase, ProbBBC computes a strategy $\strategy$ of the approximate MDP $\MDP$ that maximizes the satisfaction probability $p_{\compose{\MDP}{\strategy},\Spec}$ of $\Spec$ on $\MDP$ using PMC.
We note that the strategy $\strategy$ may not maximize the satisfaction probability $p_{\compose{\SUT}{\strategy},\Spec}$ of $\Spec$ on the SUT $\SUT$
due to the potential deviation of the approximate MDP $\MDP$ from the SUT $\SUT$.\@

In the validation phase, ProbBBC conducts statistical hypothesis testing to check whether the obtained $\strategy$ also gives $p_{\compose{\SUT}{\strategy},\Spec}$ of $\varphi$ for $\SUT$.
%
%
ProbBBC obtains execution traces by executing the SUT $\SUT$ with $\strategy$ multiple times.
By calculating how many traces satisfy $\Spec$, ProbBBC obtains an estimation $p$ of $p_{\compose{\SUT}{\strategy},\Spec}$.\footnote{We remark that estimation of $p_{\compose{\SUT}{\strategy},\Spec}$ by executing $\compose{\SUT}{\strategy}$ is not trivial since $\SUT$ is a black-box system, whereas $\strategy$ requires the information on the current state of the system in each step.
To address this challenge, we execute $\compose{\SUT}{\strategy}$ using $\Mealy$ as a scaffold; we explain this technique in \cref{section:comparison_MDP_SUT_with_strategy}.}
Using this estimated value $p$ of $p_{\compose{\SUT}{\strategy},\Spec}$, ProbBBC conducts statistical hypothesis testing with the null hypothesis $p_{\compose{\SUT}{\strategy},\Spec} = p_{\compose{\MDP}{\strategy},\Spec}$.
If $p_{\compose{\SUT}{\strategy},\Spec} \ne p_{\compose{\MDP}{\strategy},\Spec}$ is established through the hypothesis test, ProbBBC tries to construct an execution trace $\trace$ witnessing the deviation of the approximate MDP $\MDP$ from $\SUT$; if such $\trace$ is found, ProbBBC feeds it back to the learning phase to refine the approximate MDP $\MDP$.

If ProbBBC fails to find such $\trace$, it resorts to random sampling in an attempt to detect a deviation between the approximate MDP $\MDP$ and the SUT $\SUT$.  ProbBBC randomly synthesizes input sequences, feeds them to both the approximate MDP $\MDP$ and the SUT $\SUT$, and compares the probabilistic distribution of their outputs.
If ProbBBC finds a witness of their deviation, it feeds the witness back to the learning phase.
Otherwise, ProbBBC deems the approximate MDP $\MDP$ converged to the SUT $\SUT$ and returns the satisfaction probability $p_{\compose{\SUT}{\strategy},\Spec}$ of the specification $\Spec$ on the SUT $\SUT$, which is already obtained in the previous phase.

Moreover, we optimized the validation phase using the information used for MDP learning. For example, we immediately terminate the validation phase and return to the learning phase when some assumptions regarding the estimation of transition probabilities prove to be incorrect.

We implemented a prototype tool of ProbBBC.\@ We conducted experiments to evaluate the performance of ProbBBC using benchmarks related to CPS or IoT scenarios, which are also used in previous papers on MDP learning~\cite{DBLP:journals/fac/TapplerA0EL21} or MDP testing via learning~\cite{DBLP:journals/fmsd/AichernigT19}.
Results from our experiment indicate that the satisfaction probability $p_{\compose{\SUT}{\strategy},\Spec}$ estimated by ProbBBC is usually close to the true value.
Moreover, ProbBBC tends to estimate a better probability than an existing approach~\cite{DBLP:journals/fmsd/AichernigT19}, especially when the observability in the SUT $\SUT$ is limited.

\paragraph{Summary of the contributions}
Our contributions are summarized as follows.

\begin{itemize}
 \item We propose probabilistic black-box checking (ProbBBC) by combining active MDP learning, probabilistic model checking, and statistical hypothesis testing.
 \item We optimized the validation phase using the information for MDP learning.
 \item Our experiment results suggest that our approach outperforms an existing approach.
\end{itemize}

\paragraph{Outline}
We show some related approaches in \cref{section:related_work}.
After recalling the preliminaries in \cref{sec:preliminary}, we introduce probabilistic black-box checking (ProbBBC) in \cref{sec:pbbc}.
We show the experimental evaluation in \cref{sec:experiments}, and conclude in \cref{sec:conclusion}.

\section{Related Work}\label{section:related_work}

\begin{table}[tbp]
 \caption{Testing methods for black-box systems with external inputs and probabilistic branching. PMC stands for probabilistic model checking, and SMC stands for statistical model checking.}%
 \label{table:related_work}
 \ShortVersion{\scriptsize}\LongVersion{\small}
 \begin{tabular}{c c c c}
  \toprule
  & (Near-) optimal strategy synthesis & Probability computation & MDP learning \\
  \midrule
  \textbf{Ours} & PMC & PMC & Active~\cite{DBLP:journals/fac/TapplerA0EL21}\\
  \cite{DBLP:journals/fmsd/AichernigT19} & PMC & PMC & Passive~\cite{DBLP:journals/ml/MaoCJNLN16}\\
  \cite{DBLP:conf/sefm/2014w,DBLP:journals/sttt/DArgenioLST15} & random sampling & SMC & N/A\\
  \cite{DBLP:conf/atva/BrazdilCCFKKPU14} & delayed Q-learning~\cite{DBLP:conf/icml/StrehlLWLL06} & SMC & N/A\\
  \bottomrule
 \end{tabular}
\end{table}

\cref{table:related_work} summarizes testing methods for black-box systems with external inputs and probabilistic branching.
These methods are categorized into two approaches: PMC-based and SMC-based approaches.
ProbBBC is categorized as the PMC-based approach.
In the PMC-based approach, we learn an MDP $\MDP$ that approximates the SUT $\SUT$, then compute an optimal strategy $\strategy$ for the learned MDP $\MDP$ along with its corresponding probability using PMC.\@
When the learned MDP $\MDP$ closely approximates the SUT $\SUT$, we can expect the strategy $\strategy$ to be near-optimal for the SUT $\SUT$ as well.

\citet{DBLP:journals/fmsd/AichernigT19} propose another method of this approach, which we denote as \baselineMethod{}.
The main difference from ProbBBC is in the MDP learning algorithm: ProbBBC \emph{actively} learns an MDP with \LstarMDP{}~\cite{DBLP:journals/fac/TapplerA0EL21}, whereas \baselineMethod{} \emph{passively} learns an MDP with \Aalergia{}~\cite{DBLP:journals/ml/MaoCJNLN16}. An active MDP learning algorithm adaptively samples execution traces for learning within the learning algorithm. In contrast, a passive MDP learning algorithm requires externally constructed execution traces.
In~\cite{DBLP:journals/fmsd/AichernigT19}, they use an $\varepsilon$-greedy algorithm to sample execution traces.  This algorithm primarily samples traces using the the best strategy at the time but also employs a random input with probability $\varepsilon$. Due to its greedy sampling approach, the method often gets stuck in a suboptimal strategy, as we observe experimentally in \cref{sec:experiments}.
In contrast, ProbBBC tends outperform in finding a near-optimal strategy because \LstarMDP{} samples execution traces so that the transition function from each state can be identified, and the obtained execution traces tend to cover broader behavior.
Another less significant difference lies in the properties each method supports: ProbBBC supports safety LTL, whereas \baselineMethod{} is limited to reachability properties.

\emph{Statistical model checking (SMC)}~\cite{DBLP:journals/tomacs/AghaP18} estimates the satisfaction probability of a specification through random sampling of execution traces. Because of the sampling-based approach of SMC, it can handle black-box systems.
SMC is primarily for purely probabilistic systems with no external inputs (\eg{} Markov chains rather than MDPs), and when we apply SMC to systems with external inputs, we need to synthesize a (near-) optimal strategy to determine inputs.
In~\cite{DBLP:conf/sefm/LegayST14,DBLP:journals/sttt/DArgenioLST15}, strategies are constructed by random sampling with a concise encoding of each strategy.
In~\cite{DBLP:conf/atva/BrazdilCCFKKPU14}, strategies are constructed by delayed Q-learning~\cite{DBLP:conf/icml/StrehlLWLL06}.
These methods, which directly learn and apply a strategy on the SUT, require the current state of the SUT to be observabile.\@
In contrast, ProbBBC indirectly uses a strategy via the approximate MDP, and the current state can be unobservable.

Schematically, ProbBBC is based on \emph{black-box checking (BBC)}~\cite{DBLP:conf/forte/PeledVY99}.
ProbBBC is a quantitative extension of BBC in the following sense: 
i) It handles systems with \emph{probabilistic} transitions, which is more general than \emph{deterministic} transitions supported by BBC;\@
ii) It tries to return a \emph{strategy} to \emph{maximize} the satisfaction probability of the given property $\Spec$, which is a quantitative generalization of returning an \emph{input sequence} \emph{violating} $\Spec$.
Furthermore, BBC for a deterministic system $\SUT$ against an LTL formula $\Spec$ is reducible to ProbBBC for $\SUT$ against $\neg \Spec$ by checking if the maximum satisfaction probability of $\neg \Spec$ is 1.

\section{Preliminaries}\label{sec:preliminary}

For a set $S$, we denote its power set by $\powerset{S}$, the set of finite sequences of $S$ elements by $S^*$, and the set of the infinite sequence of $S$ by $S^{\omega}$.
For an infinite sequence $s = s_0, s_1, \dots \in S^\omega$ and $i, j \in \mathbb{N}$ where $i \leq j$, we denote the subsequence $s_i, s_{i+1}, \dots , s_j \in S^*$ by $s[i, j]$ and $s_i, s_{i+1}, \dots \in S^\omega$ by $s[i, \infty]$.
We write $s \cdot s'$ for the concatenation of a finite sequence $s \in S^*$ and an infinite sequence $s' \in S^\omega$ of $S$.
A \emph{trace} $\trace$ is an alternating sequence of inputs and outputs (i.e., $\trace \in \OUTPUT \times (\INPUT \times \OUTPUT)^*$.)
We denote the set of probability distributions over $S$ by $\Dist[S]$: for any $\mu : S \to [0, 1]$ in $\Dist[S]$, $\sum_{s \in S} \mu(s) = 1$ holds.
For sets $X$ and $Y$, we denote $X \subseteqfin Y$ if $X \subseteq Y$ holds and $X$ is finite.

\subsection{Model of systems and specification logic}

Here, we briefly review some notions related to (probabilistic) model checking.
See, \eg{}~\cite{DBLP:conf/sfm/KwiatkowskaNP07} for a more formal reasoning of probabilities using measure theory.

\subsubsection{Mealy machine}

\begin{definition}[Mealy machine]
  A (deterministic) Mealy machine is a 5-tuple $\Mealy = (Q, \INPUT, \OUTPUT, q_0, \Delta)$, where $Q$ is the finite set of states, $\INPUT$ and $\OUTPUT$ are the input and output alphabets, $q_0 \in Q$ is the initial state, and $\Delta : (Q \times \INPUT) \to (Q \times \OUTPUT)$ is the transition function.
We write $\mathcal{L}(\mathcal{M}) \subseteq (\INPUT \times \OUTPUT)^\omega$ for the \emph{language} of a Mealy machine $\mathcal{M}$ defined as follows:
\[
  \mathcal{L}(\mathcal{M}) \coloneqq \left\{ (\action_0, \observation_0),(\action_1, \observation_1), \dots \mid \exists q_1,q_2, \ldots \in Q^{\omega}, \forall i \in \mathbb{N}.\,  \Delta(q_i, \action_i) = (q_{i+1}, \observation_i) \right\}.
\]
\end{definition}
For $\sigma \in (\INPUT)^\omega$ and a Mealy machine $\mathcal{M}$, we write $\mathcal{M}(\sigma)$ for the output obtained by feeding $\sigma$ to $\mathcal{M}$.
More precisely, $\mathcal{M}(\sigma)$ for $\sigma = \action_0,\action_1,\dots \in (\INPUT)^\omega$ is defined as $\observation_0,\observation_1,\dots \in (\OUTPUT)^\omega$ such that $(\action_0,\observation_0),(\action_1,\observation_1),\dots \in \mathcal{L}(\mathcal{M})$.
Notice that this notation is well-defined since $\mathcal{M}$ is deterministic.
We write $\mathcal{L}^{\FIN}(\mathcal{M})$ for the \emph{finite} language of $\mathcal{M}$ defined as the set of finite prefixes of $\mathcal{L}(\mathcal{M})$: $\mathcal{L}^{\FIN}(\mathcal{M}) \coloneqq \{\sigma \in (\INPUT \times \OUTPUT)^* \mid \exists \sigma' \in (\INPUT \times \OUTPUT)^\omega.\, \sigma \cdot \sigma' \in \mathcal{L(M)}\}$.
%

\subsubsection{Markov decision process}

We use \emph{Markov decision processes (MDPs)} for modeling systems that exhibit stochastic behavior.
\begin{definition}[Markov decision process]
  A Markov decision process (MDP) is a 6-tuple $\MDP = (Q, \INPUT, \OUTPUT, q_0, \Delta, L)$, where $Q$ is the finite set of states, $\INPUT$ and $\OUTPUT$ are input and output alphabet, $q_0 \in Q$ is the initial state, $\Delta\colon (Q \times \INPUT) \to \Dist[Q]$ is the probabilistic transition function, and $L\colon Q \to \OUTPUT$ is the labeling function.
  A \emph{path} $\rho$ of $\MDP$ is an element of $Q \times (\INPUT \times Q)^*$.
  We write $\mathit{Path}_{\MDP}$ for the set of all paths of MDP $\MDP$.
  We say MDP $\MDP$ is \emph{deterministic} iff for any $q \in Q$ and for any $\action \in \INPUT, q', q'' \in Q$, $\Delta(q, \action)(q') > 0$ and $\Delta (q, \action)(q'') > 0$ implies $q' = q''$ or $L(q') \neq L(q'')$.
\end{definition}
In this paper, we let $\OUTPUT = \powerset{\AP}$ where $\AP$ is a set of relevant atomic propositions, and the labeling function $L$ returns the propositions that hold at the given state.

An execution of an MDP $\MDP = (Q, \INPUT, \OUTPUT, q_0, \Delta, L)$ can be seen as an interaction with a system and a controller.
An execution of $\MDP$ starts at the initial state $q_0$.
Suppose the system is at state $q$.
The controller chooses an input $\action \in \INPUT$ to the system; then, the system's state changes based on the probability distribution $\Delta(q,\action)$.
This interaction between a system and its controller is expressed by a path $\rho$ of states and inputs starting at the initial state $q_0$, \ie{} $\rho = q_0, \action_1, q_1, \action_2, q_2, \dots \action_n, q_n$; at each state $q_k$, the input $\action_{k+1}$ is chosen by the controller, and the next state is chosen probabilistically according to the distribution $\Delta(q_k,\action_{k+1})$.

A path $\rho = q_0, \action_1, q_1, \action_2, q_2, \dots \action_n, q_n$ induces a trace $\trace$ of the MDP $\MDP$; $\trace = \observation_0, \action_1, \observation_1, \action_2,\observation_2, \dots \action_n, \observation_n$, where for each $k \in \{0,1,\dots,n\}, \observation_k = L(q_k)$.
%
If the MDP $\MDP$ is deterministic, the path $\rho$ corresponding to a trace $\trace$ is uniquely identified.

The intuition of the interaction between a system and its controller is formalized by modeling the controller as a \emph{strategy}.
\begin{definition}[Strategy]
For an MDP $\MDP = (Q, \INPUT, \OUTPUT, q_0, \Delta, L)$,
a \emph{strategy} for $\MDP$ is a function $\strategy\colon \mathit{Path}_{\MDP} \to \Dist[\INPUT]$.
\end{definition}
A strategy $\strategy$ can be seen as a controller that probabilistically chooses an input $\action \in \INPUT$ to be fed to $\MDP$ from a path $\rho \in \mathit{Path}_{\MDP}$ that represents the execution of the system so far based on the probability distribution $s(\rho)$.

An interacting system composed of an MDP $\MDP$ and a strategy $\strategy$ is modeled as a discrete-time Markov chain (DTMC)~\cite{DBLP:conf/sfm/ForejtKNP11}.
We write $\compose{\MDP}{\strategy}$ for this DTMC.
\ks{Removed the formalization and the transition probability of DTMC.  Let's write it once we recognize they are needed.}




A strategy is \emph{finite-memory} if its choice of inputs depends only on the current state and a finite mode updated by each input and state, not on the entire path. Any finite-memory deterministic strategy can be encoded by a Mealy machine $(M, Q \times \INPUT, \INPUT, m_0, E)$, where the choice of input $\action \in \INPUT$ depends only on $M$ and $Q$ (\ie{} for any $m \in M$, $q \in Q$, and $\action,\action' \in \INPUT$, the second elements of $E(m, (q, \action))$ and $E(m, (q, \action'))$ are the same).
For an MDP $\MDP = (Q, \INPUT, \OUTPUT, q_0, \Delta, L)$ and a finite-memory deterministic strategy $\strategy$ of $\MDP$ encoded by a Mealy machine $(M, Q \times \INPUT, \INPUT, m_0, E)$, $\MDP_{\strategy}$ is a \emph{finite-state} DTMC with the state space $Q \times M \times \INPUT$, initial state $(q_0, m_0, \action_0)$, where $\action_0 \in \INPUT$ is such that $(m, \action_0) = E(m_0,(q_0, \action))$ for some $m \in M$ and $\action \in \INPUT$, and the transition probability $P((q_{k+1}, m_{k+1}, \action_{k+1}) \mid (q_k, m_k, \action_k))$ from $(q_k, m_k, \action_k) \in Q \times M \times \INPUT$ to $(q_{k+1}, m_{k+1}, \action_{k+1}) \in Q \times M \times \INPUT$ is 
$\Delta(q_k, L(m_k, (q_k, \action_k)))(q_{k+1})$ if $L(m_k, (q_k, \action_k)) = (m_{k+1}, \action_{k+1})$, and otherwise $0$.


\subsubsection{Linear temporal logic (LTL)}

We use a probabilistic extension of the \emph{linear temporal logic (LTL)}~\cite{DBLP:conf/focs/Pnueli77} for specifying the temporal properties of the behavior of a system.
\begin{definition}[Syntax of LTL]
  \label{definition:LTL}
  For a finite set $\AP$ of atomic propositions, the syntax of \emph{linear temporal logic} is defined as follows,
  where $p \in \AP$ and $i, j \in \N \cup \{ \infty \} $ satisfying $i \leq j$
  \[
    \varphi , \psi \Coloneqq \top \mid p \mid \neg \varphi \mid \varphi \vee \psi \mid \Next \varphi \mid \varphi \Until[[i, j)] \psi.
  \]
\end{definition}

An LTL formula $\varphi$ asserts a property on an infinite sequence $\pi \in (\powerset{\AP})^\omega$ of subsets of $\AP$.
Intuitively, $\top$ holds for any $\pi$; $p$ holds if the first element of $\pi$ includes $p$; $\neg \varphi$ holds if $\varphi$ does not hold for $\pi$; $\varphi \lor \psi$ holds if either $\varphi$ or $\psi$ holds for $\pi$; and $\Next \varphi$ holds if $\varphi$ holds for $\pi[1,\infty]$.
The formula $\varphi \Until[[i, j)] \psi$ asserts that (1) $\psi$ becomes true in the future along with $\pi$ that arrives within $l \in [i,j)$ steps and (2) $\phi$ continues to hold until $\psi$ becomes true.
More precisely, $\varphi \Until[[i, j)] \psi$ holds if $\psi$ holds for some $\pi' = \pi[l,\infty]$ such that $l \in [i,j)$ and $\phi$ holds for $\pi[m,\infty]$ for any $m \in [0,l)$.

The semantics of LTL formulas is defined by the following satisfaction relation $(\pi, k) \models \varphi$ that represents that $\varphi$ holds for $\pi[0,k]$.
For an infinite sequence $\pi$, an index $k$, and an LTL formula $\varphi$, $(\pi, k) \models \varphi$ intuitively stands for ``$\pi$ satisfies $\varphi$ at $k$''.

\begin{definition}
  [Semantics of LTL]
  \label{definition:LTLsemantics}
  For an LTL formula $\Spec$, an infinite sequence $\pi = \pi_0, \pi_1, \dots \in (\powerset{\AP})^\omega$\LongVersion{ of subsets of atomic propositions}, and $k \in \mathbb{N}$,
  we define the satisfaction relation $(\pi, k) \models \Spec$ as follows.
 \begin{gather*}
  (\pi, k) \models \top \qquad
  (\pi, k) \models p \iff p \in \pi_k \qquad
  (\pi, k) \models \neg \varphi \iff (\pi, k) \nvDash \varphi \\
  (\pi, k) \models \varphi \lor \psi \iff (\pi, k) \models \varphi \lor (\pi, k) \models \psi \qquad
  (\pi, k) \models \Next \varphi \iff (\pi, k + 1) \models \varphi \\
  (\pi, k) \models \varphi \Until[[i, j)] \psi \iff \exists l \in [k + i, k + j).\, (\pi, l) \models \psi 
             \land \forall m \in \{k, k + 1, \dots, l-1\}.\, (\pi, m) \models \varphi
 \end{gather*}
  If we have $(\pi, 0) \models \varphi$, we denote $\pi \models \varphi$.
\end{definition}

We use the following syntax sugars of LTL formulas defined as follows.
\begin{gather*}
  \bot \equiv \neg \top, \quad
  \varphi \wedge \psi \equiv \neg ((\neg \varphi) \vee (\neg \psi)), \quad
  \varphi \rightarrow \psi \equiv (\neg \varphi) \vee \psi, \quad
  \Evt_{[i, j)}\varphi \equiv \top \Until[[i, j)] \varphi, \quad \\
  \Glb_{[i, j)} \varphi \equiv \neg (\Evt_{[i, j)} \neg \varphi),
  \varphi \Until \psi \equiv \varphi \Until[[0, \infty)] \psi, \quad
  \Evt \varphi \equiv \Evt_{[0, \infty)} \varphi, \quad
  \Glb \varphi \equiv \Glb_{[0, \infty)} \varphi
\end{gather*}
The intuition of $\bot$, $\varphi \land \psi$, and $\varphi \rightarrow \psi$ should be clear.
The formula $\Evt_{[i,j)} \varphi$ asserts that ``$\varphi$ holds eventually between $i$-step and $j$-step future''; the formula $\Glb_{[i,j)} \varphi$ stands for ``$\varphi$ holds globally between $i$-step and $j$-step future''.
We also introduce $\Until$, $\Evt$, and $\Glb$ for the unbounded versions of $\Until[[i,j)]$, $\Evt_{[i,j)}$, and $\Glb_{[i,j)}$.

Given a Mealy machine $\mathcal{M}$ and an LTL formula $\varphi$, we write $\mathcal{M} \models \varphi$ if $\mathcal{M}(\sigma) \models \varphi$ for any $\sigma \in (\INPUT)^\omega$.
We write $\mathcal{M} \not\models \varphi$ if $\mathcal{M} \models \varphi$ does not hold.

\emph{Safety} LTL is a subclass of LTL that consists of the LTL formulas expressing safety properties, whose violation can be witnessed by a \emph{finite} sequence.
\begin{definition}[Safety]\label{def:safety}
  An LTL formula $\varphi$ is \emph{safety} if, for any infinite sequence $\pi \in {(\powerset{\AP})}^{\omega}$ satisfying $\pi \nvDash \varphi$, there is $i \in \N$ such that 
  for any\LongVersion{ infinite sequence} $\pi' \in {(\powerset{\AP})}^\omega$, we have $\pi[0, i]\cdot \pi' \nvDash \varphi$.
\end{definition}

\subsubsection{Quantitative probabilistic model checking}

\emph{Probabilistic model checking (PMC)}~\cite{DBLP:reference/mc/BaierAFK18} is a method to verify MDPs against LTL formulas.
We use \emph{quantitative} PMC, which computes the maximum (or minimum) satisfaction probability of the given LTL formula $\Spec$ on the MDP $\MDP$.

\defProblem{Quantitative probabilistic model checking}{%
An MDP $\MDP = (Q, \INPUT, \OUTPUT, q_0, \Delta, L)$ and an LTL formula $\Spec$}{%
Computes the maximum satisfaction probability $\max_{s \in \mathbf{Sched}_{\MDP}} \Prob_{\MDP, s} (\varphi)$ of $\Spec$ on $\MDP$, where $\mathbf{Sched}_{\MDP}$ is the set of strategies of $\MDP$ and $\Prob_{\MDP, \strategy} (\Spec)$ is the probability of an execution of the DTMC $\MDP_s$ satisfying $\Spec$.
}

For any MDP $\MDP$ and an LTL formula $\Spec$, there is a finite-memory deterministic strategy $\strategy$ maximizing the satisfaction probability of $\Spec$ on $\MDP$~\cite{DBLP:conf/atva/KwiatkowskaP13}.
A probabilistic model checker, \eg{} PRISM~\cite{DBLP:conf/cav/KwiatkowskaNP11}, takes an MDP $\MDP$ and an LTL formula $\Spec$, and computes a finite-memory deterministic strategy $\strategy$ that maximizes the probability of $\compose{\MDP}{\strategy}$ satisfying $\Spec$; it returns $\strategy$ as well as the probability $p$.

\subsection{Active learning of automata and MDPs}\label{subsection:active_automata_learning}

\todo{Move the figure if the layout is ugly}
\emph{Active automata learning} is a class of algorithms to construct an automaton
through interactions between the \emph{learner} and a \emph{teacher}.
In the \Lstar{} algorithm~\cite{DBLP:journals/iandc/Angluin87},\LongVersion{ the best known active automata learning algorithm,}
the learner constructs the minimum DFA $\A$ over $\Alphabet$ recognizing the target language $\targetLg \subseteq \Alphabet^*$
using \emph{membership} and \emph{equivalence} questions to the teacher.
In a membership question, the learner asks if a word $\word \in \Alphabet^*$ is a member of $\targetLg$, \ie{} $\word \in \targetLg$.
Membership questions are used to obtain sufficient information to construct a candidate DFA $\hypothesisA$.
In an equivalence question, the learner asks if the candidate DFA $\hypothesisA$ recognizes the target language $\targetLg$, \ie{} $\Lg(\hypothesisA) = \targetLg$.
If we have $\Lg(\hypothesisA) \neq \targetLg$,
the teacher returns a word $\cex$ satisfying $\cex \in \Lg(\hypothesisA) \setdiff \targetLg$ as a witness of $\Lg(\hypothesisA) \neq \targetLg$,
where $\Lg(\hypothesisA) \setdiff \targetLg$ is the symmetric difference between $\Lg(\hypothesisA)$ and $\targetLg$,
\ie{} $\Lg(\hypothesisA) \setdiff \targetLg = (\Lg(\hypothesisA) \setminus \targetLg) \cup (\targetLg \setminus \Lg(\hypothesisA))$.
Equivalence questions are used to decide if the learning process can be finished.

\reviewer{2}{The specification of the observation table in lines 333--339 and Figure 4 are hard to understand.}
\begin{wrapfigure}[9]{r}{0pt}
 \scriptsize
 \begin{tabular}{c|c c}
  & $\emptyword$ & b \\\hline
  $\emptyword$& $\top$ & $\bot$\\
  a & $\bot$ & $\bot$\\
  b & $\bot$ & $\top$\\ \hline
  aa & $\top$ & $\bot$\\
  ab & $\bot$ & $\bot$\\
  ba & $\bot$ & $\bot$\\
  bb & $\top$ & $\bot$\\
 \end{tabular}
 \caption{An observation table in \Lstar{}.}%
 \label{figure:observation_table}
\end{wrapfigure}
The \Lstar{} algorithm uses a 2-dimensional array called an \emph{observation table} to maintain the information obtained from the teacher.
\cref{figure:observation_table} illustrates an observation table during learning 
$\targetLg = \{\word \in {(\text{ab})}^* \mid (\text{\# of a in } \word) \mod 2 = (\text{\# of b in } \word) \mod 2 = 0 \}$.
The rows and columns of an observation table are indexed by finite sets of words $\PrefixSet \cup \ExtPrefixSet$ and $\SuffixSet$.
In \cref{figure:observation_table}, $\PrefixSet$ is displayed above the horizontal line:
we have $\PrefixSet = \{ \emptyword, \mathrm{a}, \mathrm{b}\}$,
$\ExtPrefixSet \setminus \PrefixSet = \{\mathrm{aa}, \mathrm{ab}, \mathrm{ba}, \mathrm{bb}\}$, and
$\SuffixSet = \{ \emptyword, \mathrm{b}\}$.
The cell indexed by $(\prefix, \suffix) \in (\PrefixSet \cup \ExtPrefixSet) \times \SuffixSet$ shows if the concatenation $\prefix \cdot \suffix \in \Alphabet^*$ is a member of $\targetLg$.
For example, the cell indexed by $(\text{a}, \text{b})$ shows $\text{ab} \not\in \targetLg$.

\begin{wrapfigure}[8]{r}{0pt}
  \begin{tikzpicture}[shorten >=1pt,scale=0.8,every node/.style={transform shape},every initial by arrow/.style={initial text={}}]
  \node[initial,state,accepting] (l0) at (0,0) [align=center]{$\emptyword$};
  \node[state] (l1) at (3.0,0) [align=center]{a};
  \node[state] (l2) at (1.5,1.5) [align=center]{b};

  \path[->] 
  (l0) edge [bend left=10] node[above] {a} (l1)
  (l0) edge [bend left=10] node[above left] {b} (l2)
  (l1) edge [bend left=10] node[below] {a} (l0)
  (l1) edge [loop above] node {b} (l1)
  (l2) edge node[above right] {a} (l1)
  (l2) edge [bend left=10] node[below right] {b} (l0)
  ;
  \end{tikzpicture}
  \caption{Candidate DFA $\hypothesisA$ constructed from the observation table in \cref{figure:observation_table}.}%
  \label{figure:candidate_DFA}
\end{wrapfigure}
A DFA can be constructed from an observation table satisfying certain conditions.
\cref{figure:candidate_DFA} shows the DFA constructed from the observation table in \cref{figure:observation_table}.
In the DFA construction, a state is constructed for each unique row.
For example, from the observation table in \cref{figure:observation_table},
three states are constructed, corresponding to the rows $(\top, \bot)$, $(\bot, \bot)$, and $(\bot, \top)$.
A state is accepting if the cell indexed by $(\prefix, \emptyword)$ is $\top$, where $\prefix$ is the index of a row corresponding to the state.
The successor of the state corresponding to the row indexed by $\prefix$ with $a \in \Alphabet$ is
the state corresponding to the row indexed by $\prefix \cdot a$.
For example, from the observation table in \cref{figure:observation_table},
the successor of the state corresponding to $(\top, \bot)$, which is indexed by $\emptyword$, with a is $(\bot, \bot)$, which is indexed by a.
For the above construction, the observation table must satisfy the following two conditions.
Note that the observation table in \cref{figure:observation_table} satisfies these conditions.
\begin{description}
 \item[closedness] For each row indexed by $\prefix \in \ExtPrefixSet \setminus \PrefixSet$, there is a row indexed by some $\prefix' \in \PrefixSet$ with the same contents.
 \item[consistency] For each pair of rows indexed by $\prefix,\prefix' \in \PrefixSet$ if their contents are the same, for any $a \in \Alphabet$, the rows indexed by $\prefix \cdot a$ and $\prefix' \cdot a$ also have the same contents.
\end{description}
\ks{Remark that we use only the hypothesis construction phase of \LstarMDP{}.  The equivalence checking part is replaced with our strategy-guided comparison.}
\begin{LongVersionBlock} 
 \begin{algorithm}[tb]
  \caption{Outline of the candidate-generation procedure.}%
  \label{algorithm:candidateConstruction}
  \DontPrintSemicolon{}
  \newcommand{\myCommentFont}[1]{\texttt{\footnotesize{#1}}}
  \SetCommentSty{myCommentFont}
  \SetKwFunction{FAskMembershipQ}{askMembershipQuestion}
  \SetKwFunction{FConstructCandidate}{constructCandidateAutomaton}
  \While(\tcp*[f]{Candidate generation phase}) {Observation table is not closed or consistent} {\label{algorithm:L*style:ClosedAndConsistent}\label{algorithm:L*style:AutomatonConstruction}
    update $\PrefixSet$ or $\SuffixSet$ and fill the observation table with membership questions\;
  }\label{algorithm:L*style:ClosedAndConsistentEnd}
    \Return{$\FConstructCandidate(\PrefixSet, \SuffixSet)$}\label{algorithm:L*style:AutomatonConstructionEnd}\;
\end{algorithm}
\end{LongVersionBlock}

\begin{LongVersionBlock} 
\begin{algorithm}[tb]
  \caption{Outline of the equivalence checking procedure.}%
  \label{algorithm:equivalenceChecking}
  \DontPrintSemicolon{}
  \newcommand{\myCommentFont}[1]{\texttt{\footnotesize{#1}}}
  \SetCommentSty{myCommentFont}
  \SetKwFunction{FAskMembershipQ}{askMembershipQuestion}
  \SetKwFunction{FConstructCandidate}{constructCandidateAutomaton}
    \If(\tcp*[f]{Equivalence checking phase with an equivalence question}) {$\targetLg = \Lg(\hypothesisA)$} {\label{algorithm:L*style:EquivalenceTesting}
      \KwReturn{} $\mathrm{OK}(\hypothesisA)$
    } \Else{
      \KwLet{} $\cex$ \KwBe{} an evidence of $\targetLg \neq \Lg(\hypothesisA)$, \ie{} $\cex \in \targetLg \setdiff \Lg(\hypothesisA)$\label{algorithm:L*style:WitnessReturn}\;
      \KwReturn{} $\mathrm{CEx}(\cex)$
    }\label{algorithm:L*style:EquivalenceTestingEnd}
\end{algorithm}
\end{LongVersionBlock}

\begin{algorithm}[tb]
  \ShortVersion{\footnotesize}
  \caption{Outline of the \Lstar{} algorithm for active automata learning.}%
  \label{algorithm:L*styleLearning}
  \DontPrintSemicolon{}
  \newcommand{\myCommentFont}[1]{\texttt{\footnotesize{#1}}}
  \SetCommentSty{myCommentFont}
  \SetKwFunction{FAskMembershipQ}{askMembershipQuestion}
  \SetKwFunction{FConstructCandidate}{constructCandidateAutomaton}
  $\PrefixSet \gets \{ \emptyword\}$;\,\,$\SuffixSet \gets \{ \emptyword\}$\;
  \While{$\top$} {
    \While (\tcp*[f]{Candidate generation phase}) {Observation table is not closed or consistent} {\label{algorithm:L*style:ClosedAndConsistent}\label{algorithm:L*style:AutomatonConstruction}
      update $\PrefixSet$ or $\SuffixSet$ and fill the observation table with membership questions\;
    }\label{algorithm:L*style:ClosedAndConsistentEnd}
    $\hypothesisA \gets \FConstructCandidate(\PrefixSet, \SuffixSet)$\label{algorithm:L*style:AutomatonConstructionEnd}\;
    \If (\tcp*[f]{Equivalence checking phase with an equivalence question}) {$\targetLg = \Lg(\hypothesisA)$} {\label{algorithm:L*style:EquivalenceTesting}
      \KwReturn $\hypothesisA$
    } \Else {
      \KwLet{} $\cex$ \KwBe{} an evidence of $\targetLg \neq \Lg(\hypothesisA)$, \ie{} $\cex \in \targetLg \setdiff \Lg(\hypothesisA)$\label{algorithm:L*style:WitnessReturn}\;
      add prefixes of $\cex$ to $\PrefixSet$ and fill the observation table with membership questions\;\label{algorithm:L*style:AddCexToP}
    }\label{algorithm:L*style:EquivalenceTestingEnd}
  }  
\end{algorithm}

\cref{algorithm:L*styleLearning} outlines the \Lstar{} algorithm.
The \Lstar{} algorithm consists of two phases: candidate generation and equivalence checking phases.
%
In the candidate generation phase\LongVersion{, which is outlined in \cref{algorithm:candidateConstruction}}, the learner increases $\PrefixSet$ and $\SuffixSet$ and fills the observation table using membership questions until the observation table becomes closed and consistent (\crefrange{algorithm:L*style:ClosedAndConsistent}{algorithm:L*style:ClosedAndConsistentEnd}).
Once the observation table becomes closed and consistent, a hypothesis DFA $\hypothesisA$ is constructed (\cref{algorithm:L*style:AutomatonConstructionEnd}).
Then, in the equivalence checking phase\LongVersion{ outlined in \cref{algorithm:equivalenceChecking}},
the learner checks if $\hypothesisA$ recognizes the target language using an equivalence question (\cref{algorithm:L*style:EquivalenceTesting}).
If we have $\targetLg = \Lg(\hypothesisA)$, $\hypothesisA$ is returned.
Otherwise, the teacher returns an evidence $\cex$ of $\targetLg \neq \Lg(\hypothesisA)$ (\cref{algorithm:L*style:WitnessReturn}), and the prefixes of $\cex$ are added to $\PrefixSet$ (\cref{algorithm:L*style:AddCexToP}).

\ShortVersion{\begin{wrapfigure}[8]{r}{0pt}}\LongVersion{\begin{figure}[tbp]}
 \centering
 \ShortVersion{\scriptsize}
  \begin{tabular}{c|c c}
   & a & b \\\hline
   p& $\{\mathrm{p} \mapsto 15, \mathrm{q} \mapsto 15\}$ & $\{\mathrm{q} \mapsto 25\}$\\
   paq & $\{\mathrm{p} \mapsto 8, \mathrm{q} \mapsto 12\}$ & $\{\mathrm{p} \mapsto 4, \mathrm{q} \mapsto 6\}$\\ \hline
   pap & $\{\mathrm{p} \mapsto 5, \mathrm{q} \mapsto 5\}$ & $\{\mathrm{q} \mapsto 12\}$\\
   pbp & $\emptyset$ & $\emptyset$\\
   pbq & $\{\mathrm{p} \mapsto 4, \mathrm{q} \mapsto 6\}$ & $\{\mathrm{p} \mapsto 2, \mathrm{q} \mapsto 3\}$ \\
   paqap & $\{\mathrm{p} \mapsto 3, \mathrm{q} \mapsto 3\}$ & $\{\mathrm{q} \mapsto 4\}$\\
   paqaq & $\{\mathrm{p} \mapsto 2, \mathrm{q} \mapsto 3\}$ & $\{\mathrm{p} \mapsto 2, \mathrm{q} \mapsto 3\}$ \\
   paqbp & $\{\mathrm{p} \mapsto 3, \mathrm{q} \mapsto 3\}$ & $\{\mathrm{q} \mapsto 4\}$\\
   paqbq & $\{\mathrm{p} \mapsto 2, \mathrm{q} \mapsto 3\}$ & $\{\mathrm{p} \mapsto 2, \mathrm{q} \mapsto 3\}$ \\
  \end{tabular}
  \caption{An observation table in \LstarMDP{} with $\INPUT = \{\text{a}, \text{b}\}$ and $\OUTPUT = \{\text{p}, \text{q}\}$.}%
  \label{figure:observation_table_MDP}
\ShortVersion{\end{wrapfigure}}\LongVersion{\end{figure}}
The \LstarMDP{} algorithm~\cite{DBLP:journals/fac/TapplerA0EL21} is an extension of the \Lstar{} algorithm for active MDP learning.
It learns a deterministic MDP using similar questions and an observation table.
In the \LstarMDP{} algorithm, the teacher maintains a prefix-closed multiset $\TracePool$ of traces, which contains all the information obtained by executing the system under learning and used to answer queries.
\cref{figure:observation_table_MDP} illustrates an observation table in the \LstarMDP{} algorithm.
The following summarizes the major differences in the observation table.

\begin{itemize}
 \item To learn an MDP with inputs $\INPUT$ and outputs $\OUTPUT$, the indices are modified: the row indices are $\PrefixSet \cup (\PrefixSet \cdot \INPUT \cdot \OUTPUT) \subseteqfin \OUTPUT \cdot {(\INPUT \cdot \OUTPUT)}^*$ and the column indices are $\SuffixSet \subseteqfin \INPUT \cdot {(\OUTPUT \cdot \INPUT)}^*$.
 \item The cell indexed by $(\prefix,\suffix)$ represents a function $T_{\prefix,\suffix}\colon\OUTPUT \to \N$ mapping each output to the frequency of its appearance in $\TracePool$ after $\prefix \cdot \suffix$. 
       For instance, in \cref{figure:observation_table_MDP}, the cell indexed by (p, a) shows that we observed each of ``pap'' and ``paq'' for 15 times.
\end{itemize}

The notion of closedness and consistency are also updated to statistically compare the rows with a \emph{Hoeffding bound}~\cite{hoeffding1963probability}.
To maintain such an observation table, the membership question is replaced with the following questions:
i) Given a trace $\trace \in \OUTPUT \cdot {(\INPUT \cdot \OUTPUT)}^*$, it returns the function mapping $\observation \in \OUTPUT$ to the frequency $\TracePool(\trace \cdot \observation)$ of $\trace \cdot \observation$ in the multiset $\TracePool$;
ii) Given a trace $\trace \in \OUTPUT \cdot {(\INPUT \cdot \OUTPUT)}^*$, it returns if $\TracePool$ contains sufficient information to estimate the output distribution after $\trace$;
iii) It asks the teacher to refine $\TracePool$ by sampling traces rarely appearing in $\TracePool$.

\ShortVersion{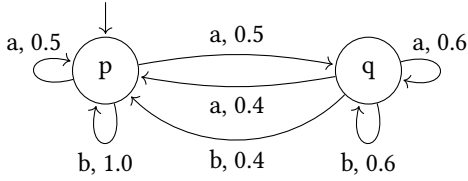
\begin{wrapfigure}[8]{r}{0pt}}\LongVersion{\begin{figure}[tbp]}
  \LongVersion{\centering}
  \ShortVersion{\begin{tikzpicture}[shorten >=1pt,scale=0.8,every node/.style={transform shape},every initial by arrow/.style={initial text={}}]}
  \LongVersion{\begin{tikzpicture}[shorten >=1pt\ShortVersion{,scale=0.8},every node/.style={transform shape},every initial by arrow/.style={initial text={}}]}
  \node[initial above,state] (l0) at (0,0) [align=center]{p};
  \node[state] (l1) at (3.5,0) [align=center]{q};

  \path[->] 
  (l0) edge [bend left=10] node[above] {a, 0.5} (l1)
  (l0) edge [loop left] node[above=0.1] {a, 0.5} (l0)
  (l0) edge [loop below] node {b, 1.0} (l0)
  (l1) edge [bend left=10] node[below] {a, 0.4} (l0)
  (l1) edge [loop right] node[above=0.1] {a, 0.6} (l1)
  (l1) edge [bend left=45] node[below] {b, 0.4} (l0)
  (l1) edge [loop below] node[below] {b, 0.6} (l1)
  ;
  \end{tikzpicture}
  \caption{Candidate MDP constructed from the observation table in \cref{figure:observation_table_MDP}. The label at each state represents the output.}%
  \label{figure:candidate_MDP}
\ShortVersion{\end{wrapfigure}}\LongVersion{\end{figure}}
The MDP construction from an observation table is also similar to the DFA construction in the \Lstar{} algorithm.
\cref{figure:candidate_MDP} shows the MDP constructed from the observation table in \cref{figure:observation_table_MDP}.
State construction is based on the comparison of rows by Hoeffding bound.
Transitions are constructed by estimating the probability distribution of the successors using the contents of the cells.

\subsection{Black-box checking}\label{subsection:bbc}

\begin{figure}[t]
  \centering
  \tikzset{
  >={Latex[width=2mm,length=2mm]},
  state/.style={
      rectangle,
      draw=black, very thick,
      minimum height=2em,
      inner sep=4pt,
      align=center,
    },
  }
  \ShortVersion{\begin{tikzpicture}[node distance=1.6cm,align=center,scale=0.75,every node/.style={transform shape},every initial by arrow/.style={initial text={}}]}
  \LongVersion{\begin{tikzpicture}[node distance=1.6cm,align=center,every node/.style={transform shape},every initial by arrow/.style={initial text={}}]}
    \node[state, initial above] (n1)
    {Learn a Mealy machine $\Mealy$\\ that approximates $\SUT$};
    \node[state, below of=n1, node distance=2.7cm] (n2)
    {Verify if\\$\Mealy \models \varphi$ by\\model checking};
    \node[state, right of=n2, xshift=3.5cm] (n3)
    {Test if\\$\SUT \nvDash \Spec$ is\\witnessed by $\sigma$};
    \node[state, left of=n2, xshift=-3.5cm] (n4)
    {Check if\\$\mathcal{M} \simeq \Mealy$ by\\equivalence testing};
    \node[node distance=0.75cm,below=of n3] (n5) {$\SUT \nvDash \Spec$ witnessed by $\sigma$};
    \node[node distance=0.75cm,below=of n4] (n6) {Deems $\SUT \models \Spec$};

    \node[above of=n1, yshift=-0.74cm, xshift=-1.9cm] (n7) {(A)};
    \node[above of=n2, yshift=-0.5cm, xshift=-1.2cm] (n8) {(B)};
    \node[above of=n3, yshift=-0.5cm, xshift=-0.7cm] (n9) {(C)};
    \node[above of=n4, yshift=-0.5cm, xshift=-1.2cm] (n10) {(D)};

    \draw[->] (n1) -- node[fill=white,yshift=0.2cm]
    {Learned Mealy machine $\Mealy$} (n2);
    \path[->] (n2) edge
        node[above] {$\Mealy \nvDash \Spec$}
        node[below,align=center]{witnessed by $\sigma$} (n3);
    \path[->] (n2) edge node[above] {$\Mealy \models \Spec$} (n4);
    \draw[->] (n3.north) -- node[fill=white,xshift=0cm,yshift=0.3cm]
    {No.\\($\SUT \neq \Mealy$ is witnessed by $\sigma$)} ++(0,1.9) -- (n1.east);
    \draw[->] (n4.north) -- node[fill=white,xshift=0cm,yshift=0.3cm]
    {$\SUT \neq \Mealy$ \ is\\witnessed by $\sigma$} ++(0,1.9) -- (n1.west);
    \path[->] (n3) edge node[right,yshift=0.1cm] {Yes} (n5);
    \path[->] (n4) edge node[right,yshift=0.1cm] {Deems $\SUT = \Mealy$} (n6);
  \end{tikzpicture}
  \caption{The workflow of black-box checking.}%
  \label{fig:bbc-flow}
\end{figure}
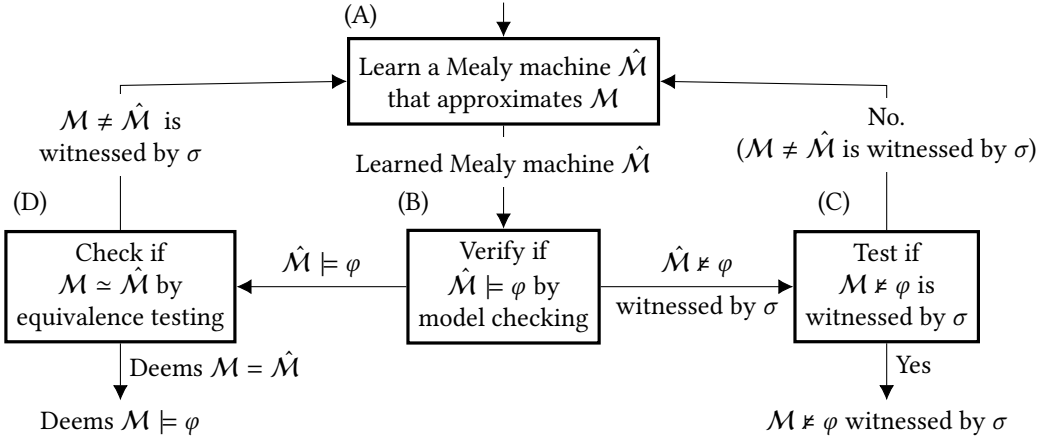

\emph{Black-box checking (BBC)}~\cite{DBLP:conf/forte/PeledVY99} is a method for testing a black-box system $\mathcal{M}$ against its specification $\Spec$.
BBC takes the following inputs: 
i) a deterministic black-box system $\SUT$ that takes $\sigma \in {(\INPUT)}^\omega$ as input and outputs $\SUT(\sigma) \in \powerset{\AP}^\omega$ and 
ii) a safety LTL formula $\Spec$.
The input to and the output from $\SUT$ are both streams.
The aim of BBC is to find a counterexample $\sigma' \in {(\INPUT)}^{*}$ that witnesses the violation of $\varphi$: For any $\sigma'' \in {(\INPUT)}^\omega$, the counterexample $\sigma'$ satisfies $\SUT(\sigma' \cdot \sigma'') \not\models \Spec$.
If such a counterexample is not found, BBC reports so.

BBC consists of learning, synthesis, and validation phases.
BBC iterates the learning phase to obtain a Mealy machine $\Mealy$ that approximates the behavior of $\SUT$, the synthesis phase to synthesize a witness $\sigma$ of $\Mealy \not\models \varphi$, and the validation phase to check whether $\sigma$ is also the true counterexample to $\SUT \models \varphi$.
%
By combining automata learning and model checking, BBC guides the testing procedure efficiently instead of randomly generating various test inputs and expecting some of them to falsify the specification $\Spec$.

\cref{fig:bbc-flow} illustrates the workflow of BBC.\@
We explain each component in the figure below.
(The heading of each item corresponds to a box in \cref{fig:bbc-flow}.)
\begin{description}
\item[(A)] Using the candidate generation phase of a variant of the \Lstar{} algorithm, BBC obtains a Mealy machine $\Mealy$ approximating the behavior of $\SUT$.
\item[(B)] Then, BBC verifies $\Mealy$ against the specification $\Spec$ using model checking.
\item[(C)] Suppose the model checker in (B) asserts $\Mealy \not\models \Spec$.
Let $\sigma \in {(\INPUT)}^*$ be a counterexample that witnesses $\Mealy \not\models \Spec$; notice that such $\sigma$ exists in ${(\INPUT)}^*$ because $\varphi$ is a safety LTL formula.
This $\sigma$ may not be a valid counterexample for $\mathcal{M}$ because $\Mealy$ is merely an approximation of $\mathcal{M}$.
BBC checks whether $\sigma$ is a valid counterexample also for $\SUT$ by feeding $\sigma$ to $\SUT$ and checking whether $\SUT(\sigma) \not\models \Spec$ holds.
If $\SUT(\sigma) \not\models \varphi$, then $\sigma$ is a valid counterexample for $\SUT$.
Otherwise, although $\sigma$ does not witness $\SUT \not\models \Spec$, it does witness the behavioral difference between $\SUT$ and $\Mealy$.
Therefore, BBC adds $\sigma$ to the data used by the learning phase and jumps to Step (A).
Then, the learning procedure progresses to learn more precise $\Mealy$.
\item[(D)] Suppose the model checker used in Step (B) successfully verifies $\Mealy \models \Spec$.
Then, BBC tests the equivalence between $\SUT$ and $\Mealy$, typically by randomly generating many elements of ${(\INPUT)}^*$ and feeding them to $\mathcal{M}$ and $\Mealy$.
If this step finds $\sigma \in {(\INPUT)}^*$ such that $\SUT(\sigma) \ne \Mealy(\sigma)$, then this $\sigma$ witnesses the behavioral difference between $\SUT$ and $\Mealy$; therefore, BBC adds $\sigma$ to the data used by the learning phase and jumps to Step (A).
Otherwise, BBC reports that no counterexample is found deeming $\Mealy$ is equivalent to $\SUT$.
\end{description}

In \cref{fig:bbc-flow}, (A) corresponds to the learning phase, (B) corresponds to the synthesis phase, and (C---D) correspond to the validation phase.
Notice that the validation phase (C) and (D) as a whole checks the equivalence of $\SUT$ and $\Mealy$.
Instead of checking the equivalence only by the inefficient random test in (D), BBC first checks a necessary condition of the equivalence: $\SUT(\sigma) \not\models \varphi$ for $\sigma$ such that $\Mealy(\sigma) \not\models \varphi$ obtained in (C).
BBC then proceeds to (D) only if this check in (C) passes.

\section{Probabilistic black-box checking}\label{sec:pbbc}

We present \emph{probabilistic black-box checking (ProbBBC)}.
ProbBBC is an extension of BBC in \cref{subsection:bbc} for stochastic systems.
\mw{I elaborated the assumptions}
We assume that the system under test (SUT) is a \emph{black-box MDP} with finite states, \ie{}
there is an underlying MDP $(Q, \INPUT, \OUTPUT, q_0, \Delta, L)$ with $|Q| < \infty$ representing the SUT $\SUT$, but we only know $\INPUT$ and $\OUTPUT$.
We assume that the MDP is deterministic, which intuitively requires sufficient observability to distinguish the next states.
We also assume that the SUT is \emph{executable}, \ie{}
we can perform the following probabilistic operations on the \emph{unobservable} current state $q \in Q$ of the SUT $\SUT$.

\begin{itemize}
 \item We can reset the current state $q$\LongVersion{ of the SUT $\SUT$} to the initial state $q_0$ and obtain the output $L(q_0)$.
 \item For an input $\action \in \INPUT$, we can update the current state $q$\LongVersion{ of the SUT $\SUT$} to $q' \in Q$ according to the distribution $\Delta(q, \action)$ and obtain the output $L(q')$ after the transition.
\end{itemize}

We believe that these assumptions are acceptable in many usage scenarios. For example, in reinforcement learning, Q-learning~\cite{DBLP:journals/ml/WatkinsD92} assumes that the environment is a black-box MDP that is finite, executable, and fully observable, which is more restrictive than ours.

Along with the stochastic extension of the SUT, we also change the problem to synthesize an optimal strategy rather than finding an input violating the given LTL formula.
The problem we approximately solve with ProbBBC is summarized as follows.

\defProblem{Optimal strategy synthesis}{%
A black-box MDP $\SUT$ and an LTL formula $\Spec$}{%
Find a strategy $\strategy$ maximizing the satisfaction probability $p_{\compose{\SUT}{\strategy},\Spec}$ of $\Spec$ on $\SUT$}

\subsection{Overview of ProbBBC}\label{subsec:pbbc-workflow}

\cref{fig:pbbc-flow} outlines ProbBBC.\@
As the conventional BBC does, ProbBBC consists of the learning phase ((A) in \cref{fig:pbbc-flow}), the synthesis phase ((B) in \cref{fig:pbbc-flow}), and the validation phase ((C---E) in \cref{fig:pbbc-flow}).
Given a black-box system $\SUT$ and an LTL formula $\varphi$, ProbBBC applies (1) the candidate generation phase of \LstarMDP{} in~\cref{subsection:active_automata_learning} to construct an MDP $\MDP$ that approximates $\SUT$ in the learning phase ((A) in \cref{fig:pbbc-flow}) and (2) quantitative probabilistic model checking~\cite{DBLP:reference/mc/BaierAFK18} to synthesize $\strategy$ that maximizes the probability of $\compose{\MDP}{\strategy}$ satisfying $\Spec$ in the synthesis phase ((B) in \cref{fig:pbbc-flow}).

The goal of the validation phase is the same as that of the conventional BBC:\@ to check that $\SUT$ and $\MDP$ are equivalent to each other.
To this end, the validation phase first checks the necessary condition of the equivalence: $p_{\compose{\SUT}{\strategy}, \varphi}$ is the same as $p_{\compose{\MDP}{\strategy},\varphi}$, where $p_{\compose{\SUT}{\strategy}, \Spec}$ and $p_{\compose{\MDP}{\strategy}, \Spec}$ are the satisfaction probabilities of $\Spec$ on $\compose{\SUT}{\strategy}$ and $\compose{\MDP}{\strategy}$ ((C) and (D) in \cref{fig:pbbc-flow}).
We call this procedure \emph{strategy-guided comparison} of $\SUT$ and $\MDP$.
If this necessary condition seems to be satisfied, the validation phase applies the equivalence checking phase of \LstarMDP{} to check the equivalence of $\SUT$ and $\MDP$ \emph{without} guided by the strategy $\strategy$ ((E) in \cref{fig:pbbc-flow}).

The strategy-guided comparison is inspired by the witness checking of the conventional BBC ((C) in \cref{fig:bbc-flow}).
However, unlike the witness checking of the conventional BBC, which can check whether the counterexample obtained by model checking is a true counterexample for $\SUT$ by executing $\SUT$ only once, there are following challenges in the strategy-guided comparison of ProbBBC.\@
\begin{itemize}
  \item ProbBBC must estimate the \emph{probability} $p_{\compose{\SUT}{\strategy},\varphi}$ of $\SUT_\strategy$ satisfying $\varphi$.  A one-shot execution of $\SUT_\strategy$ is not enough to estimate this probability.
  \item ProbBBC needs to decide whether $p_{\compose{\SUT}{\strategy},\varphi}$ differs from $p_{\compose{\MDP}{\strategy},\varphi}$ obtained in the synthesis phase.
  However, since an estimation of $p_{\compose{\SUT}{\strategy},\varphi}$ is a probabilistic variable that is not guaranteed to be the true probability, simply comparing an obtained estimation with $p_{\compose{\MDP}{\strategy},\varphi}$ is not enough.
\end{itemize}

To address these challenges, ProbBBC first executes $\SUT_{\strategy}$ multiple times and estimates the probability $p_{\compose{\SUT}{\strategy},\varphi}$ from the samples ((1) of (C) in \cref{fig:pbbc-flow}).
Then, ProbBBC conducts statistical hypothesis testing with the null hypothesis $p_{\compose{\SUT}{\strategy},\varphi} = p_{\compose{\MDP}{\strategy},\varphi}$ ((2) of (C) in \cref{fig:pbbc-flow}).
If $p_{\compose{\SUT}{\strategy},\varphi} \ne p_{\compose{\MDP}{\strategy},\varphi}$ is established from the hypothesis testing, ProbBBC constructs a trace that witnesses $\SUT \ne \MDP$; this trace is added to the data that the learning phase uses ((D) in \cref{fig:pbbc-flow}).

If the strategy-guided comparison fails to find an evidence of $\SUT \neq \MDP$, 
we compare $\SUT$ and $\MDP$ using the equivalence checking phase of the \LstarMDP{} algorithm in \cref{subsection:active_automata_learning} ((E) in \cref{fig:pbbc-flow}).
In our implementation,
among various choices of the input construction (\eg{} the W-method~\cite{DBLP:journals/tse/Chow78}),
we use uniform random sampling, \ie{} each input $\action \in \INPUT$ is sampled from the uniform distribution over $\INPUT$ and the sampling stops with a certain stop probability.
Moreover, the stop probability is reduced when we fail to find an evidence of $\SUT \neq \MDP$ so that any length of inputs are eventually sampled.
Such an adaptive change of stop probability makes the parameters robust to the complexity of $\SUT$.

To enhance ProbBBC, we utilize the information obtained during the strategy-guided comparison also in the learning phase.
%
In the strategy-guided comparison, ProbBBC adds the traces obtained by sampling $\SUT_{\strategy}$ to the observation table used in the learning phase.
The closedness and the consistency of the observation table are periodically checked; if the observation table turns out to not-closed or inconsistent, then the execution of the validation phase is stopped, and the learning phase starts to learn new $\MDP$ with the new observation table.
The interruption also happens when the validation phase discovers a trace that is impossible in $\MDP$.
These optimizations are achieved by (1) sharing the multiset $\TracePool$ of traces mentioned in \cref{subsection:active_automata_learning} between the learning and validation phases and (2) adding the traces discovered during the validation phase to $\TracePool$.

\begin{figure}[tbp]
  \centering
  \tikzset{
  >={Latex[width=2mm,length=2mm]},
  state/.style={
      rectangle,
      draw=black, very thick,
      minimum height=2em,
      inner sep=4pt,
      align=center,
    },
  }
  \ShortVersion{\begin{tikzpicture}[node distance=1.6cm,align=center,scale=0.70,every node/.style={transform shape},every initial by arrow/.style={initial text={}}]}
  \LongVersion{\begin{tikzpicture}[node distance=1.6cm,align=center,scale=0.85,every node/.style={transform shape},every initial by arrow/.style={initial text={}}]}
    \node[state, initial above] (n1)
    {Generate a candidate MDP\\ that approximates $\mathcal{M}$};
    \node[node distance=0.2cm,above right=of n1, xshift=0.2cm,yshift=-0.3cm] (learning_phase) {\underline{\textbf{Learning phase}}};
    \node[state, node distance=1.0cm, below=of n1] (n2)
    {Probabilistic model checking of\\ $\MDP$ against $\Spec$};
    \node[node distance=1.2cm,below=of learning_phase] (synthesis_phase) {\underline{\textbf{Synthesis phase}}};
    \node[state, node distance=2.0cm, below=of n2, text depth=3.5cm, text width=7cm] (n3)
    {Check if $\strategy$ reveals the difference\\between $\SUT$ and $\MDP$ with respect to $\Spec$};
    \node[node distance=2.5cm,below=of synthesis_phase] (validation_phase) {\underline{\textbf{Validation phase}}};
    \node[dashed,draw,thick,node distance=1.4cm,below=of n2,text depth=7.3cm, text width=15.5cm] (validation_phase_rectangle) {};
    \node[state, node distance=3.0cm, below=of n2, text width=5.5cm] (n10)
    {(1) Sampling traces from $\compose{\SUT}{\strategy}$};
    \node[state, node distance=1.0cm, below=of n10, text width=5.5cm] (n11)
    {(2) Statistical hypothesis testing\\ to establish $p_{\compose{\MDP}{\strategy},\Spec} \neq \bar{p}_{\compose{\SUT}{\strategy},\Spec}$};
    \node[state, right of=n11, xshift=4.5cm, yshift=-2cm] (n5)
    {Try to construct\\a trace $\trace$ witnessing\\$\SUT \neq \MDP$};
    \node[state, left of=n11, xshift=-3.8cm, yshift=-2cm] (n6)
    {Check if\\$\SUT \simeq \MDP$\\ by equivalence testing};
    \node[node distance=1.25cm, below=of n6] (n7) {the probability $\bar{p}_{\compose{\SUT}{\strategy},\Spec}$\\and the strategy $\strategy$};

    \node[above of=n1, yshift=-0.74cm, xshift=-1.9cm] (na) {(A)};
    \node[above of=n2, yshift=-0.74cm, xshift=-2.4cm] (nb) {(B)};
    \node[above=of n3, yshift=-1.6cm, xshift=-3.7cm] (nc) {(C)};
    \node[above of=n5, yshift=-0.5cm, xshift=-1.3cm] (ne) {(D)};
    \node[above of=n6, yshift=-0.5cm, xshift=-1.1cm] (nf) {(E)};

    \draw[->] (n1) -- node[fill=white] {candidate MDP $\MDP$} (n2);
    \draw[->] (n2) -- node[fill=white,yshift=0.25cm]
    {the probability $p_{\compose{\MDP}{\strategy},\Spec}$ of satisfying $\Spec$\\on $\MDP$ with a strategy $\strategy$} (n3);
    \path[->] (n10) edge node[fill=white,yshift=0.1cm] 
    {probability $\bar{p}_{\compose{\SUT}{\strategy},\Spec}$ of satisfying $\Spec$ on $\compose{\SUT}{\strategy}$} (n11);
    \draw[->] (n11.south) -- ++(-1,0) -- ++(0,-1.5) -- 
    node[below, xshift=0.5cm]{Deems $p_{\compose{\MDP}{\strategy},\Spec} = \bar{p}_{\compose{\SUT}{\strategy},\Spec}$} (n6);
    \draw[->] (n11.south) -- ++(1,0) -- ++(0,-1.5) -- 
    node[below, xshift=-0.5cm]{$p_{\compose{\MDP}{\strategy},\Spec} \neq \bar{p}_{\compose{\SUT}{\strategy},\Spec}$ is established} (n5);

    \draw[->] (n5.north)  -- node[fill=white,xshift=0cm,yshift=-2.5cm]
    {Success\\($\mathcal{M} \neq \MDP$ \ is\\witnessed by $\sigma$)} ++(0,9)  -- (n1.east);
    \draw[->] (n5.south) -- ++(0,-0.7) -- node[above] {Failure} ++(-7,0) -- (n6);
    \draw[->] (n6.north) -- node[fill=white,yshift=-2.5cm] 
    {$\SUT \neq \MDP$ \ is\\witnessed by $\trace$} ++(0,9) -- (n1.west);
    \draw[->] (n6) edge node[fill=white,yshift=0.2cm] {Deems $\SUT = \MDP$} (n7);
  \end{tikzpicture}
  \caption{The workflow of probabilistic black-box checking.}%
  \label{fig:pbbc-flow}
\end{figure}
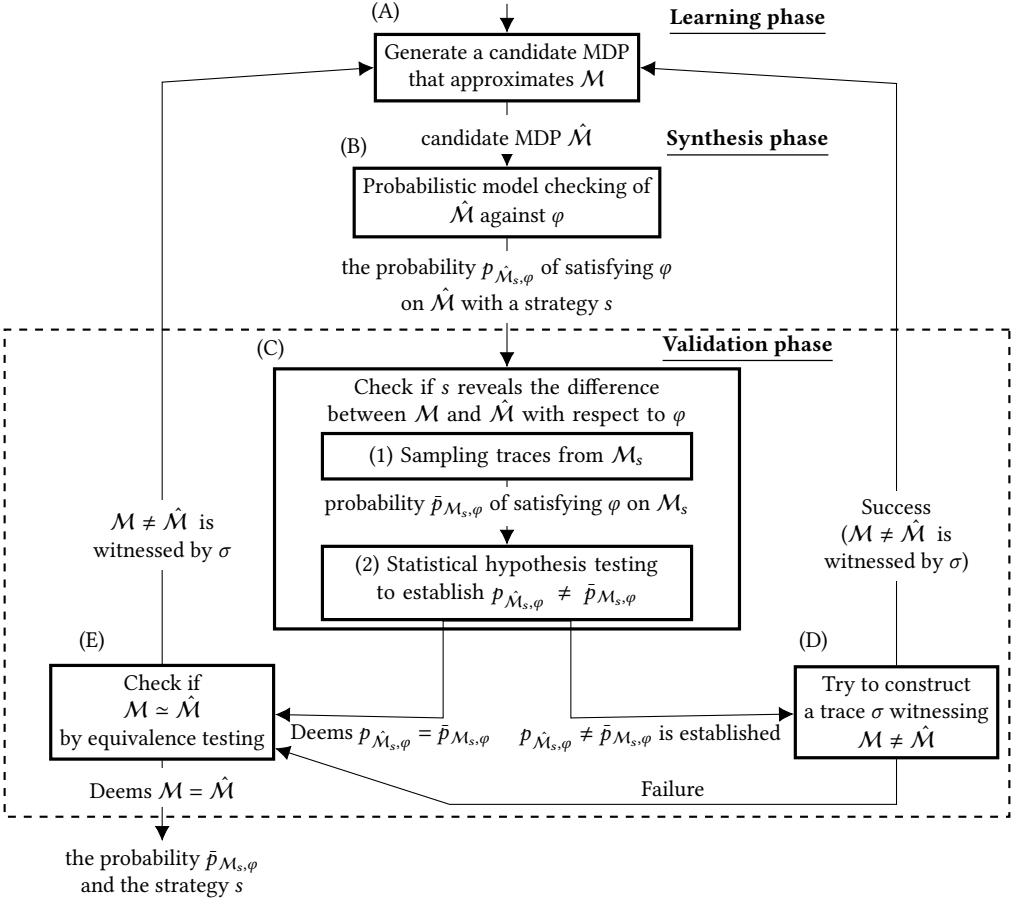

Notice that the combination of the synthesis and the validation phases ((B---E) in \cref{fig:pbbc-flow}) can be considered as the equivalence checking phase of the \LstarMDP{} algorithm because, overall, it tries to find an evidence of $\SUT \neq \MDP$.
Therefore, properties of \LstarMDP{} (such as convergence) shown in~\cite{DBLP:journals/fac/TapplerA0EL21} also holds for ProbBBC with additional discussion on the strategy-guided comparison.
As a corollary of the convergence, we have the correctness of ProbBBC,
which we show in \cref{subsec:pbbc-correctness}.

\subsection{Detail of the strategy-guided comparison in the validation phase}\label{subsec:strategy_guided_comparison}

\reviewer{1}{Please revise Section 4.2 to make it easier to follow, ideally with some running examples.}

In the following, we explain the detail of the validation phase of ProbBBC, focusing on the strategy-guided comparison.
This part is of the largest technical novelty among the three phases.

\subsubsection{Comparison of $\SUT$ and $\MDP$ with a strategy}%
\label{section:comparison_MDP_SUT_with_strategy}

\begin{algorithm}[tbp]
  \ShortVersion{\footnotesize}
  \caption{Comparison of $\SUT$ and $\MDP$ with a strategy $\strategy$ and an LTL formula $\Spec$.}%
  \label{algorithm:SamplingPhaseOfSMC}
  \newcommand{\myCommentFont}[1]{\texttt{\footnotesize{#1}}}
  \SetCommentSty{myCommentFont}
  \DontPrintSemicolon{}
  \SetKwFunction{FSamplingPhaseOfSMC}{CompareWithStrategy}
  \SetKwFunction{FSampleSingleTrace}{SampleSingleTrace}
  \SetKwFunction{FreqNumOfSample}{RequiredNumberOfSample}
  \SetKwFunction{FUpdateObservation}{UpdateObservationTable}
  \SetKwFunction{FExtendObservationTable}{ExtendObservationTableWithCEX}
  \SetKwFunction{FisPeriodicallyCheckingRound}{isPeriodicallyCheckingRound}
  \SetKwFunction{FProbabilisticModelChecking}{ProbabilisticModelChecking}
  \SetKwFunction{FStudentTesting}{StudentTesting}
  \SetKwFunction{AddPrefixes}{AddPrefixes}

  \Fn{\FSamplingPhaseOfSMC{$\MDP, \SUT, \strategy, p_{\compose{\MDP}{s},\varphi}, \Spec$}} {
    \Input{An MDP $\MDP$, the SUT $\SUT$, a strategy $\strategy$, probability $p_{\compose{\MDP}{s},\varphi}$, and a safety LTL formula $\Spec$}
    \Output{If $\MDP$ and $\SUT$ are deemed equivalent for the inputs determined by $\strategy$ with respect to $\Spec$}
    \For{$i\gets 1$ \KwTo{} $N$}{\label{algorithm:SamplingPhaseOfSMC:for}
      $\sigma \gets \FSampleSingleTrace(\mathcal{M}, s)$\label{algorithm:SamplingPhaseOfSMC:sample}\; 
      \KwAdd{} the prefixes of $\sigma$ \KwTo{} $\TracePool$\label{algorithm:SamplingPhaseOfSMC:addToTracePool}\;
      \lIf{$\sigma \models \Spec$\label{algorithm:SamplingPhaseOfSMC:evaluate}}{
        $x_i \gets 1$
      } \lElse{
        $x_i \gets 0$\label{algorithm:SamplingPhaseOfSMC:vioIncrement}
      }
    }
    $\bar{p}_{\compose{\SUT}{\strategy},\Spec} \gets \sum_{i = 0}^{k} x_i / N$\label{algorithm:SamplingPhaseOfSMC:calp};\,\, $\mathit{std}_{x} \gets \text{\texttt{StandardDeviation}}(x_1,x_2,\dots,x_N)$\;
    \tcp{Conduct Student's t-testing, where the null hypothesis is $p_{\compose{\MDP}{\strategy},\Spec} = \bar{p}_{\compose{\SUT}{\strategy},\Spec}$}
    \Return{\FStudentTesting{$p_{\compose{\MDP}{\strategy},\Spec}, \bar{p}_{\compose{\SUT}{\strategy},\Spec}, \mathit{std}_{x}, N$}}\label{algorithm:SamplingPhaseOfSMC:ret}\;
  }
\end{algorithm}

\cref{algorithm:SamplingPhaseOfSMC} outlines our algorithm to compare $\SUT$ and $\MDP$ with strategy $\strategy$, which corresponds to the box (C) in \cref{fig:pbbc-flow}.
In the loop starting from \cref{algorithm:SamplingPhaseOfSMC:for}, ProbBBC repeatedly samples traces of the SUT $\SUT$ up to the required sample size $N$.
For each iteration, we obtain a trace $\sigma$ by executing the SUT $\SUT$ (\cref{algorithm:SamplingPhaseOfSMC:sample}), recording it to the multiset $\TracePool$ shared with the learning phase (\cref{algorithm:SamplingPhaseOfSMC:addToTracePool}), and check if $\sigma$ satisfies $\Spec$ or not (\cref{algorithm:SamplingPhaseOfSMC:evaluate}).
If $\Spec$ is \emph{bounded}, \ie{} its satisfaction can be decided by traces of length $k$ for some $k \in \N$, one can sample traces of length $k$.
Otherwise, the trace length must be randomly decided.
After the sampling, we estimate the satisfaction probability $\bar{p}_{\compose{\SUT}{\strategy},\Spec}$ of $\Spec$ by $\compose{\SUT}{\strategy}$ (\cref{algorithm:SamplingPhaseOfSMC:calp}) and compare it with the satisfaction probability $p_{\compose{\MDP}{\strategy},\Spec}$ by the learned MDP (\cref{algorithm:SamplingPhaseOfSMC:ret}), which is computed by quantitative probabilistic model checking.

In executing the SUT $\SUT$ in \cref{algorithm:SamplingPhaseOfSMC:sample} in \texttt{SampleSingleTrace}, ProbBBC needs to execute $\SUT_\strategy$.
However, $\strategy$ requires the path of $\compose{\SUT}{\strategy}$ to produce an input to be fed to $\SUT$, which is not possible in this case because the sequence of the states of $\SUT$ is unknown.
We use the path of $\MDP$ instead; in executing $\compose{\SUT}{\strategy}$, ProbBBC also runs $\compose{\MDP}{\strategy}$ and maintains the corresponding path $\rho$ in the MDP $\MDP$.
Then, $\strategy(\rho)$ is fed to $\SUT$ as the next input.
Such use of $\MDP$ is justified by the convergence of $\MDP$ to $\SUT$ in the limit.
Notice that such a path $\rho = q_0, \action_1, q_1, \dots, \action_n, q_n$ is uniquely determined because 
(1) we know the initial state $q_0$ and the inputs $\action_1, \action_2, \dots, \action_n$,
(2) we can observe the output $\observation_0, \observation_1, \dots, \observation_n$ of the $\SUT$, and
(3) the successor $q_{k+1}$ in $\MDP$ is uniquely determined from the previous state $q_{k}$, the input $\action_{k+1}$, and the output $\observation_{k+1}$ because $\MDP$ is a \emph{deterministic} MDP.\@

The sample size $N$ needs to be sufficiently large so that the estimated probability $\bar{p}_{\compose{\SUT}{\strategy},\Spec}$ is close to the true probability of $\SUT_\strategy$ satisfies $\Spec$.
ProbBBC decides $N$ based on the parameters $\epsilon$ and $\delta$ specified by a user so that $\Prob(|\bar{p}_{\compose{\SUT}{\strategy},\Spec}  - p_{\compose{\SUT}{\strategy},\Spec}| \ge \epsilon) \le \delta$, where $p_{\compose{\SUT}{\strategy},\Spec}$ is the true satisfaction probability of $\Spec$ by the black-box DTMC $\compose{\SUT}{\strategy}$.
%
If $0 < p_{\compose{\SUT}{\strategy},\Spec} < 1$ holds,
using the property of Chernoff bound~\cite{Okamoto59}, we have the following property among the sample size $N$ and the parameters $\epsilon$ and $\delta$: $\delta = 2e^{-2N\epsilon^2}$; hence, ProbBBC uses $N = \left\lceil \frac{\ln(2) - \ln(\delta)}{2\epsilon^2} \right\rceil$, which is also used in the context of SMC~\cite{DBLP:conf/isola/LarsenL16}.
In our experiments, we directly fix $N$ instead of deriving it from $\delta$ and $\epsilon$.
%
%

%
%
For the comparison of $p_{\compose{\MDP}{\strategy},\Spec}$ and $\bar{p}_{\compose{\SUT}{\strategy},\Spec}$ in \cref{algorithm:SamplingPhaseOfSMC:ret}, we perform a one-sample Student's t-test~\cite{DBLP:reference/stat/KalpicHL11}, where the null hypothesis is $p_{\compose{\MDP}{\strategy},\Spec} = \bar{p}_{\compose{\SUT}{\strategy},\Spec}$.
We remark that $p_{\compose{\MDP}{\strategy},\Spec}$ is the exact maximum probability of $\MDP_\strategy$ satisfying $\Spec$, which is obtained by probabilistic model checking.
Here, the one-sample Student's t-test can be used because $N$ is large, and the binomial distribution $B(N, p_{\compose{\SUT}{\strategy},\Spec})$ is reasonably close to a normal distribution~\cite{10.5555/3134214}.

\subsubsection{Witness trace construction}\label{section:witness_trace_construction}

\begin{algorithm}[tbp]
  \ShortVersion{\footnotesize}
  \caption{Construction of the trace differentiating the SUT and an approximate MDP.}%
  \label{algorithm:searchwitness}
  \newcommand{\myCommentFont}[1]{\texttt{\ShortVersion{\scriptsize}{#1}}}
  \SetCommentSty{myCommentFont}
  \DontPrintSemicolon{}
  \SetKwFunction{FSearchWitnessOfInequivalence}{ConstructWitnessTrace}

  \Fn{\FSearchWitnessOfInequivalence{$\MDP, \TracePool$}}{
    \Input{An MDP $\MDP$ and a multiset $\TracePool$ of traces of the SUT}
    \Output{If a trace $\trace$ witnessing $\SUT \neq \MDP$ is found in $\TracePool$, returns $\trace$. Otherwise, returns $\bot$}
    \KwLet{} $\TracePool_{\mathrm{set}}$ \KwBe{} $\TracePool$ without multiplicity\label{algorithm:searchwitness:initPrefix}\;
    \While{$\TracePool_{\mathrm{set}} \neq \emptyset$\label{algorithm:searchwitness:loop_begin}} {
      $\trace \gets$ \KwPop{} one of the shortest traces \KwFrom{} $\TracePool_{\mathrm{set}}$\label{algorithm:searchwitness:sigma}\;
      \KwLet{} $\traceWithoutSuffix \cdot \otraceSingle = \trace$ such that $\traceWithoutSuffix \in {(\OUTPUT \times \INPUT)}^* $ and $\otraceSingle \in \OUTPUT$\;
      $p_{\MDP} \gets$ the probability of observing $\otraceSingle$ in $\MDP$ after $\traceWithoutSuffix$\label{algorithm:searchwitness:pot}\;
      $p_{\TracePool} \gets \TracePool(\trace) / \TracePool(\traceWithoutSuffix)$\label{algorithm:searchwitness:pphi}
      \tcp*{$p_\TracePool$ approximates the probability $p_{\SUT}$ to observe $\otraceSingle$ after $\traceWithoutSuffix$ in $\SUT$.}
      \lIf{$|p_{\MDP} - p_{\TracePool}|$ is greater than a bound}{\label{algorithm:searchwitness:SHT}
        \Return{$\sigma$}\label{algorithm:searchwitness:returnSigma}
      }
    }
    \Return{$\bot$}\label{algorithm:searchwitness:returnBot}\;
  }
\end{algorithm}

If \cref{algorithm:SamplingPhaseOfSMC} deems $\SUT \neq \MDP$, we try to construct a witnessing trace $\trace$.
The constructed witness $\trace$ is used to refine the observation table in the learning phase.

\cref{algorithm:searchwitness} outlines the witness construction.
From the multiset $\TracePool$ of the traces observed in $\SUT$, we construct the set $\TracePool_{\mathrm{set}}$ of the traces in $\TracePool$ dropping the multiplicity from $\TracePool$ (\cref{algorithm:searchwitness:initPrefix}).
Then, for each $\trace \in \TracePool_{\mathrm{set}}$, we check if $\trace$ is an evidence of $\SUT \neq \MDP$ (\crefrange{algorithm:searchwitness:loop_begin}{algorithm:searchwitness:returnSigma}) in the increasing order of the length of the traces in $\TracePool_{\mathrm{set}}$.
%
Notice that $\TracePool$ is prefix-closed, and hence $\TracePool_{\mathrm{set}}$ is also prefix-closed.
Therefore, \cref{algorithm:searchwitness} tries to find one of the shortest traces $\trace$ witnessing $\SUT \neq \MDP$.

At \cref{algorithm:searchwitness:sigma}, we pick one of the shortest traces $\trace$ from $\TracePool_{\mathrm{set}}$; at \cref{algorithm:searchwitness:pot,algorithm:searchwitness:pphi}, we compute the probability $p_{\MDP}$ of observing $\trace$ in $\MDP$ if the same input as $\trace$ is fed to $\MDP$.
This probability is computed by following the transitions of $\MDP$ according to $\trace$ and multiplying the transition probabilities in $\MDP$.
%

At \cref{algorithm:searchwitness:pphi}, we estimate the probability $p_{\SUT}$ to observe $\otraceSingle$ by feeding the prefix $\traceWithoutSuffix$ of $\trace$ such that $\trace = \traceWithoutSuffix \cdot \otraceSingle$ to $\SUT$.
Since $\SUT$ is a black-box system, we can only estimate this probability.
For the estimation, we use $\TracePool$ to approximate the trace distribution of $\SUT$ and compute the probability to observe $\trace$ after $\traceWithoutSuffix$.
Concretely, we estimate $p_{\SUT}$ by $p_{\TracePool} = \TracePool(\trace) / \TracePool(\traceWithoutSuffix)$, which is unbiased because the sampling of the outputs in $\TracePool$ follows the transition probabilities of $\SUT$.

%
%
At \cref{algorithm:searchwitness:SHT}, we compare $p_{\MDP}$ and $p_{\TracePool}$.
%
We decide whether they are different based on the criteria using a Chernoff bound mentioned in \cref{section:comparison_MDP_SUT_with_strategy}: we deem $\trace$ reveals the difference between $\MDP$ and $\SUT$ if $|p_{\MDP} - p_{\TracePool}| > \sqrt{\frac{\ln(2) - \ln(\delta)}{2 \TracePool(\traceWithoutSuffix)}}$ holds for the parameter $\delta$ in \cref{section:comparison_MDP_SUT_with_strategy}.
If $\trace$ differentiates $\MDP$ and $\SUT$, it is returned (\cref{algorithm:searchwitness:returnSigma}).
%
If we cannot find\LongVersion{ a trace differentiating $\MDP$ and $\SUT$}\ShortVersion{ such $\trace$}, we return $\bot$ (\cref{algorithm:searchwitness:returnBot}), which represents the failure of the trace construction. 


\subsubsection{Optimization using the observation table}



To enhance the entire procedure of ProbBBC, our implementation applies the following optimizations to the validation phase.
\begin{itemize}
\item For each trace $\trace$ sampled from the SUT $\SUT_\strategy$, we check if the candidate MDP $\MDP$ has a path $\rho$ corresponding to\LongVersion{ the trace} $\trace$.
If there is no such path in $\MDP$, the trace $\trace$ differentiates the SUT $\SUT$ and the candidate MDP $\MDP$.
Then, the validation phase returns such $\trace$ as a witness of $\SUT \neq \MDP$, and the learning phase starts. 
\item We also periodically update each cell of the observation table using the updated multiset $\TracePool$ of the traces obtained from $\SUT_\strategy$, and check if the observation table is still closed and consistent.
If the observation table is not closed or not consistent, we stop the validation phase and go back to the learning phase to refine $\MDP$.
\end{itemize}

\subsection{Convergence of ProbBBC}\label{subsec:pbbc-correctness}


We prove the convergence of ProbBBC by showing that MDPs learned during an execution of ProbBBC converges to one equivalent to $\SUT$ with probability $1$ given that $\SUT$ is a deterministic MDP $(Q, \INPUT, \OUTPUT, q_0, \Delta, L)$ with $|Q| < \infty$.
Therefore, the strategy obtained by ProbBBC converges to the optimal one for $\SUT$.
%
%
In this section, we assume that the equivalence checking procedure ((E) in \cref{fig:bbc-flow}) samples input, runs both $\SUT$ and $\MDP$ using the input multiple times, and subsequently compares the output distributions using hypothesis testing; the null hypothesis states that both distributions are identical.

\subsubsection{Fair-sampling assumption}

We first introduce the \emph{fair sampling} assumption for the equivalence checking procedure, which postulates that the procedure explores each observable path of the provided $\SUT$ infinitely many times with a positive probability.
To formalize this assumption, we define the set $L_{\SUT}$ of \emph{access sequences} of $\SUT$.
Essentially, $L_{\SUT}$ is the set of the observable traces corresponding to a path $\SUT$ where no state is repeated.

\begin{definition}[Access sequence]
A path $q_0,a_1,q_1,a_2,q_2,\dots,a_n,q_n$ is \emph{cycle-free} if $i \ne j$ implies $q_i \ne q_j$.
A cycle-free path $q_0,a_1,q_1,a_2,q_2,\dots,a_n,q_n$ is defined to be \emph{maximal} if, for any $a \in \INPUT$ and for any $q' \in Q$, $q' \notin \{q_0,\dots,q_n\}$ implies $\Delta(q_n,a)(q') = 0$.
A cycle-free path $q_0,a_1,q_1,a_2,q_2,\dots,a_n,q_n$ is \emph{observable} if $\Delta(q_i,a_{i+1})(q_{i+1}) > 0$ for any $i \in [0,n-1]$.
The set of \emph{access sequences} $L_{\SUT}$ is the set of traces that correspond to maximal and observable cycle-free paths of $\SUT$; concretely, $L_{\SUT}$ is defined by $\{L(q_0),a_1,L(q_1),a_2,\dots,a_n,L(q_n) \mid q_0,a_1,q_1,a_2,\dots,a_n,q_n \text{is cycle-free, observable, and maximal}\}$; as such, $L_{\SUT}$ is a finite set.
\end{definition}

\begin{definition}[Fairness assumption]
We say an equivalence testing procedure satisfies \emph{fairness assumption} if it samples every element in $L_{\SUT}$ infinitely often with a positive probability.
\end{definition}

For instance, the equivalence testing by uniform sampling (\cref{subsec:pbbc-workflow}) used in our implementation satisfies the fairness assumption.

\subsubsection{Outline of the convergence proof}

We prove the following theorem on the convergence of ProbBBC.\@

\begin{theorem}[Convergence]\label{th:convergence}
Under the fair-sampling assumption, the MDPs synthesized by (A) in an execution of ProbBBC converge to an MDP equivalent to $\SUT$ with probability $1$.
\end{theorem}

Our proof of this theorem is built upon the convergence proof of the \LstarMDP{}~\cite{DBLP:journals/fac/TapplerA0EL21}, in which the authors have shown that, under the fair-sampling condition, \LstarMDP{} converges to a correct MDP with probability $1$.

As discussed in \cref{subsec:pbbc-workflow}, the synthesis and the validation phases of ProbBBC in combination can be viewed as an enhancement to the equivalence checking of \LstarMDP{}.
Nonetheless, the convergence of ProbBBC does not directly follow from that of \LstarMDP{} due to our validation phase being based on the strategy-guided comparison, which may lead to a biased sampling of paths.
%


Crucially, we note that if we can prove each execution of the loop (A---D) terminates with probability $1$ (i.e., the loop being \emph{almost surely} terminating), we can then guarantee that the fair sampling procedure (E) is executed eventually.
Consequently, no matter the strategies synthesized by (B), every element of $L_{\SUT}$ is fairly sampled, and therefore, the convergence of ProbBBC follows.

\subsubsection{Almost-sure termination of the loop (A---D)}

Based on the discussion so far, we prove the almost-sure termination of the loop (A---D).
We first elaborate on the observation table of \LstarMDP{}.
As we mentioned in \cref{subsection:active_automata_learning}, the notion of closedness and consistency in \LstarMDP{} is based on the statistical row comparison with Hoeffding bound~\cite{hoeffding1963probability}.
Let us write $\COMPATIBLE(r,r')$ if the rows indexed by $r$ and $r'$ are deemed equal based on the comparison with the Hoeffding bound.
We also write $\EQROW(r,r')$ if the rows indexed by $r$ and $r'$ should be deemed equal based on the \emph{true} probability of $\SUT$.
We expect that $\COMPATIBLE(r,r')$ is equivalent to $\EQROW(r,r')$ after sufficiently many sampling of traces, which indeed holds.

\begin{lemma}\label{lem:compatibleIffEqrow}
If $\TracePool$ contains sufficiently many traces, then $\COMPATIBLE(r,r') \iff \EQROW(r,r')$  with probability $1$ for any rows indexed by $r$ and $r'$ in the observation table.
\end{lemma}
\begin{proof}
From Theorem 5 of \cite{DBLP:journals/fac/TapplerA0EL21}.
\end{proof}

We also use the following property of \LstarMDP{}.

\begin{lemma}\label{lem:minimality}
If $\COMPATIBLE(r,r') \iff \EQROW(r,r')$ for each $r$ and $r'$ in the row indices $\PrefixSet$ of the observation table, 
 then the MDP derived from the observation table is the smallest among ones that are consistent with the table.
\end{lemma}
\begin{proof}
From Lemma 13 of \cite{DBLP:journals/fac/TapplerA0EL21}.
\end{proof}

We show the proof sketch of the following theorem.

\begin{theorem}\label{th:terminationLoop}
The loop (A---D) of ProbBBC terminates with probability $1$ irrespective of the strategies synthesized by (B).
\end{theorem}
\noindent
\paragraph{Proof sketch.}
%
Let $\MDPi{1},\MDPi{2},\dots$ be the sequence of MDPs learned by (A) in each iteration of the execution of (A---D).
Let $s_1,s_2,\dots$ be the sequence of strategies synthesized by (C) in each iteration.
Let $\sigma_1,\sigma_2,\dots$ be the sequence of traces returned by \cref{algorithm:searchwitness} used in (D).
Let us write $\compose{\MDPi{i}}{\strategy_i}$ (resp., $\compose{\SUT}{\strategy_i}$) for the composition of $\MDPi{i}$ (resp., $\SUT$) with strategy $\strategy_i$.
After sufficiently many iterations of the loop, we can assume that $\COMPATIBLE(r,r') \iff \EQROW(r,r')$ for any indexes $r$ and $r'$ in the observation table from \cref{lem:compatibleIffEqrow}.

Suppose the procedure (D) executes \cref{algorithm:searchwitness} with $\MDPi{i}$ and $\TracePool$ as inputs.
If this execution results in $\bot$, then the loop (A---D) would terminate; therefore, assume that \cref{algorithm:searchwitness} returns $\sigma := \sigma^{-} \cdot o$, where $\sigma^{-} \in {(\OUTPUT \times \INPUT)}^*$ and $o \in \OUTPUT$.
Let $p_{\MDPi{i},\sigma^{-},o}$ be the probability of observing $o$ in $\MDPi{i}$ after $\sigma^{-}$ and $p_{\TracePool,\sigma^{-},o}$ be $\frac{\TracePool(\sigma^{-} \cdot o)}{\TracePool(\sigma^{-})}$.
Then, $|p_{\MDPi{i},\sigma^{-},o} - p_{\TracePool,\sigma^{-},o}|$ is greater than $\sqrt{\frac{\ln(2) - \ln(\delta)}{2 \TracePool(\sigma^{-})}}$.

Let $p_{\SUT,\sigma^{-},o}$ be the true probability of $\SUT$ outputting $o$ after observing $\sigma^{-}$.
Since $p_{\TracePool,\sigma^{-},o}$ is an unbiased estimation of $p_{\SUT,\sigma^{-},o}$, $\Prob(|p_{\SUT,\sigma^{-},o} - p_{\TracePool,\sigma^{-},o}| < \epsilon) > 1 - 2 e^{-2 \epsilon^2 \TracePool(\sigma^{-})}$ for an arbitrary positive real \mw{I added $\epsilon$. Is this as intended?}$\epsilon$ from the property of Chernoff bound in \cref{section:comparison_MDP_SUT_with_strategy}.
%
Notice that $|p_{\MDPi{i},\sigma^{-},o} - p_{\TracePool,\sigma^{-},o}| \le |p_{\MDPi{i},\sigma^{-},o} - p_{\SUT,\sigma^{-},o}| + |p_{\SUT,\sigma^{-},o} - p_{\TracePool,\sigma^{-},o}|$.
Therefore, $\Prob(|p_{\MDPi{i},\sigma^{-},o} - p_{\SUT,\sigma^{-},o}| > \sqrt{\frac{\ln(2) - \ln(\delta)}{2 \TracePool(\sigma^{-})}} - \epsilon) > 1 - 2 e^{-2 \epsilon^2 \TracePool(\sigma^{-})}$.
%
If we choose $\epsilon$ so that $\epsilon \in \left(\sqrt{\frac{\ln(2) - \ln(2 \delta - \delta^2))}{2 \TracePool(\sigma^{-})}}, \sqrt{\frac{\ln(2) - \ln(\delta)}{2 \TracePool(\sigma^{-})}}\right)$, we have $\Prob(|p_{\MDPi{i},\sigma^{-},o} - p_{\SUT,\sigma^{-},o}| > \sqrt{\frac{\ln(2) - \ln(\delta)}{2 \TracePool(\sigma^{-})}} - \epsilon) > (1 - \delta)^2$.
This implies that $p_{\MDPi{i},\sigma^{-},o}$ and $p_{\SUT,\sigma^{-},o}$ are indeed different with probability greater than $(1 - \delta)^2$.
%
Being $p_{\MDPi{i},\sigma^{-},o} \ne p_{\SUT,\sigma^{-},o}$ implies that, with the assumption that $\COMPATIBLE(r,r') \iff \EQROW(r,r')$ for any $r$ and $r'$, $\COMPATIBLE(\sigma^{-} \cdot o, r)$ does not hold for any $r$ in the observation table.
Therefore, adding $\sigma^{-} \cdot o$ to $\TracePool$ identifies a new state in $\MDPi{i+1}$ with probability at least $(1 - \delta)^2$.

From the above discussion, there is an infinite sequence of MDPs $\MDPi{i_1}, \MDPi{i_2}, \dots$ with probability $1$ such that $\MDPi{i_{j+1}}$ has strictly more states than $\MDPi{i_{j}}$.
However, this contradicts \cref{lem:compatibleIffEqrow} and \cref{lem:minimality}; due to these lemmas, after sufficiently many iterations of (A---D), all $\MDPi{i}$ should have less number of states than $\SUT$.
\qed{}

\section{Experimental evaluation}\label{sec:experiments}

We have developed a prototype tool implementing ProbBBC in Python\footnote{The artifact of the experiment is available on \url{https://doi.org/10.5281/zenodo.7997524}.}.
We used AALpy~\cite{DBLP:conf/atva/MuskardinAPPT21} for active MDP learning and PRISM~\cite{DBLP:conf/cav/KwiatkowskaNP11} for quantitative probabilistic model checking.

We conducted experiments to answer the following research questions.

\begin{description}
 \item[RQ1.] Does ProbBBC produce a strategy close to the optimal one?
 \item[RQ2.] Does ProbBBC produce a better strategy than the existing method~\cite{DBLP:journals/fmsd/AichernigT19}?
 \item[RQ3.] Is ProbBBC robust to the observability of the output?
 \item[RQ4.] Can ProbBBC estimate a good strategy with a small sample size?
 \item[RQ5.] Is ProbBBC sensitive to the parameters?
 \item[RQ6.] Does the strategy-guided comparison in \cref{subsec:strategy_guided_comparison} improve the performance of ProbBBC?\@
 \item[RQ7.] Is ProbBBC scalable with respect to the system's complexity?
\end{description}


\subsection{Benchmarks}

\newcommand{\PrTen}{\mathit{BAR3}}
\begin{table}[t]
 \centering
  \ShortVersion{\scriptsize}\LongVersion{\small}
 \caption{Summary of the benchmarks: the number $|Q|$ of states, the size $|\INPUT|$ of the inputs, the size $|\OUTPUT|$ of the outputs, and the tested LTL formulas.}%
 \label{table:benchmarks}
 \begin{tabular}{l c c c c}
  \toprule
  & $|Q|$ & $|\INPUT|$ & $|\OUTPUT|$ & LTL formulas \\ \midrule
    \SlotMachine{}                         & $471$   & 4 & $31$ & $\Evt_{[0,n)} \PrTen$, with $n \in \{5, 8, 11, 14, 17\}$          \\
    \SlotMachineWithSuppressedOutputs{} & $471$      & 4 & $12$ & $\Evt_{[0,n)} \PrTen$, with $n \in \{5, 8, 11, 14, 17\}$          \\
    \MQTT{}        & $62$   &   9  & $50$ & $\Evt_{[0,n)} \mathit{crash}$, with $n \in \{5, 8, 11, 14, 17\}$          \\
    \TCP{}         & $156$   &  12 & $12$ & $\Evt_{[0,n)} \mathit{crash}$, with $n \in \{5, 8, 11, 14, 17\}$          \\
    \FirstGridWorld{}     & $35$   &  4  & $7$  & $\Evt_{[0,10)} \mathit{goal}$             \\
    \SecondGridWorld{}    & $72$   &  4  & $7$  & $\Evt_{[0,13)} \mathit{goal}$              \\
    \SharedCoin{}         & $272$  &  2 & $47$ & $\Evt_{[0, n)} \mathit{finished}$, with $n \in \{14, 20\}$ \\
    \RandomGridWorld{}         & $\{16,64,100,144,196\}$  &  4 & $7$ & $(\neg \mathit{hole}) \Until{} \mathit{goal}$\\
  \bottomrule
 \end{tabular}
\end{table}

For the evaluation, we use eight benchmarks:
\SlotMachine{}, 
\SlotMachineWithSuppressedOutputs{}, 
\MQTT{}, 
\TCP{}, 
\FirstGridWorld{}, 
\SecondGridWorld{},
\SharedCoin{}, and
\RandomGridWorld{}.
\cref{table:benchmarks} summarizes them.
Each benchmark consists of an MDP and LTL formulas that are the same except for the timing parameter $n$.
\SlotMachine{} and \SlotMachineWithSuppressedOutputs{} are benchmarks with the same MDP except for observability.
\MQTT{}, \TCP{}, and \SharedCoin{} are benchmarks on communication protocols related to IoT applications.
\FirstGridWorld{} and \SecondGridWorld{} are benchmarks on\LongVersion{ controller synthesis of} a robot navigation system.
Among the eight benchmarks, \MQTT{} and \TCP{} are benchmarks aiming at testing, whereas others are on controller synthesis.
We use \RandomGridWorld{} to answer RQ7 and the others to answer the other RQs.



\subsubsection*{\SlotMachine{}}
The MDP in \SlotMachine{} is taken from~\cite{DBLP:journals/ml/MaoCJNLN16}, also used in~\cite{DBLP:journals/fac/TapplerA0EL21,DBLP:journals/fmsd/AichernigT19}.
The slot machine has three reels, which are initially ``blank''.
The player can spin each reel independently.
After a spin, each reel shows either an ``apple'' or a ``bar''.
The probability of having a ``bar'' decreases as the number of spins increases.
A player is given a maximum number $5$ of spins.
When the game is over, a prize is given depending on the reel configuration.
The player can also ``stop'' the game:
With probability 0.5, the player obtains two extra spins (up to $5$ in total);
With probability 0.5, the player finishes the game and receives the award.
Overall, there are four inputs: ``reel1'', ``reel2'', ``reel3'', and ``stop''.
The outputs show the current reel configuration (during the game) and the received prize (at the end of the game).
%
We use LTL formulas that are true if we obtain the prize $\PrTen$ for displaying ``bar'' at all three reels within specific numbers of total reels.


\subsubsection*{\SlotMachineWithSuppressedOutputs{}}
The LTL formulas in \SlotMachineWithSuppressedOutputs{} are the same as \SlotMachine{}.
The MDP in \SlotMachineWithSuppressedOutputs{} is also the same as \SlotMachine{} except for the output:
In \SlotMachineWithSuppressedOutputs{}, only the reels displaying ``bar'' are observable;
In \SlotMachine{}, the status of the reels is fully observable.
For example,
the MDP in \SlotMachineWithSuppressedOutputs{} 
has the same output for ``bar blank apple'' and ``bar apple blank'', which are distinguished in \SlotMachine{}.
We use \SlotMachineWithSuppressedOutputs{} primarily to answer RQ3.


\subsubsection*{\MQTT{}}
The MDP in \MQTT{} is taken from~\cite{DBLP:journals/fmsd/AichernigT19}, which models a broker in the MQTT protocol~\cite{MQTTv3} with stochastic failures.
Namely, the MDP models a broker crashing with probability $0.1$.
The MDP has nine inputs, \eg{} for message types, and 50 outputs for the internal state of the broker.
%
We use LTL formulas that are true if we have the output $\mathit{crash}$ representing the stochastic failure within specific numbers of messages.

\subsubsection*{\TCP{}}
The MDP in \TCP{} is also taken from~\cite{DBLP:journals/fmsd/AichernigT19}, which models a TCP server~\cite{TCP} with stochastic failures.
The probability of crashing is $0.05$.
We use LTL formulas that are true if we have the output $\mathit{crash}$ representing the stochastic failure within specific numbers of messages.

\subsubsection*{\FirstGridWorld{}}
The MDP in \FirstGridWorld{} is taken from~\cite{DBLP:journals/fmsd/AichernigT19,DBLP:journals/fac/TapplerA0EL21}, which models a robot on a grid world with probabilistic error in its movement.
For example, when a robot tries to move to the east, it may also move to the northeast or the southeast with small probabilities depending on the condition of the current position.
The MDP has four inputs (``East'', ``South'', ``West'', and ``North'') for the direction of a move and seven outputs (``Concrete'', ``Grass'', ``Wall'', ``Mud'', ``Pavement'', ``Gravel'', and ``Sand'') for the condition of the robot's current position.
%
%
We use the LTL formula $\Evt_{[0,10)} \mathit{goal}$ that is true if the robot reaches the goal position within ten steps.

\subsubsection*{\SecondGridWorld{}}

\SecondGridWorld{} is a variant of \FirstGridWorld{}, also from~\cite{DBLP:journals/fmsd/AichernigT19,DBLP:journals/fac/TapplerA0EL21}.
The state space of \SecondGridWorld{} is larger, and thus, the estimation of the MDP is harder.
The inputs and the outputs are the same as \FirstGridWorld{}.
%
We use the LTL formula $\Evt_{[0,13)} \mathit{goal}$ that is true if the robot reaches the goal position within 13 steps.

\subsubsection*{\SharedCoin{}}

The MDP in \SharedCoin{} is taken from~\cite{DBLP:journals/fac/TapplerA0EL21}, which models a randomized consensus protocol~\cite{DBLP:journals/jal/AspnesH90} with two processes.
The MDP has two inputs for the process to execute and 47 outputs, for example, for the status of the coins.
We used LTL formulas that are true if the algorithm finishes within specific numbers of messages.

\subsubsection*{\RandomGridWorld{}}

\RandomGridWorld{} is our original benchmark inspired from \FirstGridWorld{} and \SecondGridWorld{}.
\todo{Update the number of benchmarks}
\RandomGridWorld{} consists of 25 randomly generated benchmarks: we fixed five sizes of the grid world and randomly generated five MDPs for each size.
\RandomGridWorld{} is primarily used to evaluate the scalability of ProbBBC.\@
We use the LTL formula $(\neg \mathit{hole}) \Until{} \mathit{goal}$ that is true if the robot reaches the goal position without entering the hole position.

\subsection{Experiments}

To answer RQ1--RQ4, we compared the performance of \ourTool{} with \baselineMethod{}~\cite{DBLP:journals/fmsd/AichernigT19}\LongVersion{\footnote{We used the implementation of \baselineMethod{} available on \url{https://github.com/mtappler/prob-black-reach}\LongVersion{ with additional codes for experiments}.}}.
\baselineMethod{} is another testing method for black-box MDPs based on MDP learning and probabilistic model checking. The main difference from \ourTool{} is in the learning algorithm and the sampling method: \baselineMethod{} passively learns an MDP using traces sampled by an $\varepsilon$-greedy algorithm. See \cref{section:related_work} for a detailed comparison.
To answer RQ5, we compared the performance of \ourTool{} with different parameters.
To answer RQ6, we compared the performance of \ourTool{} with a variant (we call \ourToolOnlyClassic{}) of \ourTool{} without the strategy-guided comparison, \ie{} a variant such that the validation phase immediately starts equivalence testing by uniform sampling.
As a ground truth, we also compute the optimal satisfaction probability with PRISM~\cite{DBLP:conf/cav/KwiatkowskaNP11}.




For each benchmark, we ran each method (\ie{} \ourTool{} with multiple parameters, \baselineMethod{}, and \ourToolOnlyClassic{}) for 20 times.
Since each method produces a strategy for each execution of probabilistic model checking,
we have a bunch of strategies for each execution.
For each strategy,
we estimated the satisfaction probability of the LTL formula $\Spec$ by the SUT $\SUT$ using SMC with sample size 5,000.

We conducted all the experiments on a Google Cloud Platform c2-standard-4 instance (4 vCPU, 16GB RAM) running Debian 11 bullseye. We used Python 3.10.9, AALpy v.1.3.0, and PRISM version 4.7.
The parameters in the validation phase are as follows:
The sample size $N$ in \cref{algorithm:SamplingPhaseOfSMC} is $5000$;
The threshold $\Delta$ for Student's t-testing (\cref{algorithm:SamplingPhaseOfSMC:ret} of \cref{algorithm:SamplingPhaseOfSMC}) is 0.025;
The bound $\Delta'$ in the comparison of probabilities (\cref{algorithm:searchwitness:SHT} of \cref{algorithm:searchwitness}) is 0.025.
For the experiments to answer RQ5, we used all the combinations of $N \in \{2500, 5000, 10000\}$ and $\Delta = \Delta' \in \{0.01, 0.025, 0.05\}$, \ie{} we used nine parameters in total.
For the parameters of \baselineMethod{}, we used the values in~\cite{DBLP:journals/fmsd/AichernigT19}.

\cref{table:experiment_result} summarizes the maximum satisfaction probabilities of each LTL formula estimated by \ourTool{} (with the strategy-guided comparison) and \baselineMethod{} as well as the true maximum probabilities computed by PRISM.\@
\cref{figure:estimation_graph} shows the number of steps on the SUT and the largest estimated maximum probabilities before the point for some benchmarks.
\cref{table:ProbBBC_parameter_robustness} shows the mean of the maximum satisfaction probabilities of each LTL formula estimated by \ourTool{} with various choices of parameters $N$ and $\Delta$.
\cref{table:ProbBBC_vs_only-classic} summarizes the maximum satisfaction probabilities of each LTL formula estimated by \ourTool{} without the strategy-guided comparison.
\cref{figure:ProbBBC_scalability} shows the number of the states and the mean execution time to estimate a probability that is larger than $97.5\%$ of the true probability for \RandomGridWorld{}.

\begin{table}[tbp]
 \centering
 \caption{Summary of the estimated maximum satisfaction probabilities of $\Spec$ after the number of steps displayed in the column ``\# of steps''. The column $n$ shows the parameter $n$ of the LTL formulas in \cref{table:benchmarks}. The column ``PRISM'' shows the true probabilities computed by PRISM.\@ The columns ``\ourTool{}'' and ``\baselineMethod{}'' show the estimated probabilities computed by \ourTool{} and \baselineMethod{}, respectively. The columns ``mean'', ``std'', and ``min'' show the mean, the standard deviation, and the minimum value of the results. We highlight the cells if the mean of the estimated probabilities is larger than $97.5\%$ of the true probability.}%
 \label{table:experiment_result}




 \ShortVersion{\scriptsize}\LongVersion{\footnotesize}
 \begin{tabular}{llrllllllr}
  \toprule
  & \multirow[c]{2}{*}{$n$} & PRISM & \multicolumn{3}{c}{\ourTool{}} & \multicolumn{3}{c}{\baselineMethod{}} & \# of steps \\
  &  &  & mean & std & min & mean & std & min &  \\
  \midrule
  \multirow[c]{5}{*}{\SlotMachine{}} & 5 & 8.23e-02 & \goodCell{} 8.21e-02 & \goodCell{} 3.22e-03 & \goodCell{} 7.64e-02 & \goodCell{} 8.23e-02 & \goodCell{} 1.26e-03 & \goodCell{} 7.91e-02 & 15,000,000 \\
  & 8 & 3.32e-01 & \goodCell{} 3.28e-01 & \goodCell{} 6.73e-03 & \goodCell{} 3.11e-01 & 3.21e-01 & 8.29e-03 & 2.97e-01 & 15,000,000 \\
  & 11 & 4.60e-01 & \goodCell{} 4.59e-01 & \goodCell{} 8.57e-03 & \goodCell{} 4.42e-01 & 4.40e-01 & 1.56e-02 & 4.01e-01 & 15,000,000 \\
  & 14 & 4.99e-01 & \goodCell{} 4.92e-01 & \goodCell{} 1.21e-02 & \goodCell{} 4.53e-01 & 4.82e-01 & 1.30e-02 & 4.41e-01 & 15,000,000 \\
  & 17 & 5.10e-01 & \goodCell{} 5.07e-01 & \goodCell{} 9.64e-03 & \goodCell{} 4.83e-01 & 4.80e-01 & 2.37e-02 & 4.32e-01 & 15,000,000 \\
  \midrule
  \multirow[c]{5}{*}{\SlotMachineWithSuppressedOutputs{}} & 5 & 8.23e-02 & \goodCell{} 8.14e-02 & \goodCell{} 3.46e-03 & \goodCell{} 7.58e-02 & \goodCell{} 8.25e-02 & \goodCell{} 1.62e-03 & \goodCell{} 7.94e-02 & 7,000,000 \\
  & 8 & 3.32e-01 & \goodCell{} 3.31e-01 & \goodCell{} 1.00e-02 & \goodCell{} 3.17e-01 & 3.22e-01 & 5.80e-03 & 3.11e-01 & 7,000,000 \\
  & 11 & 4.60e-01 & \goodCell{} 4.61e-01 & \goodCell{} 4.24e-03 & \goodCell{} 4.52e-01 & 4.38e-01 & 1.44e-02 & 3.94e-01 & 7,000,000 \\
  & 14 & 4.99e-01 & \goodCell{} 4.95e-01 & \goodCell{} 1.11e-02 & \goodCell{} 4.62e-01 & 4.27e-01 & 4.92e-02 & 3.11e-01 & 7,000,000 \\
  & 17 & 5.10e-01 & \goodCell{} 5.09e-01 & \goodCell{} 5.46e-03 & \goodCell{} 5.01e-01 & 4.48e-01 & 4.15e-02 & 3.44e-01 & 7,000,000 \\
  \midrule
  \multirow[c]{5}{*}{\MQTT{}} & 5 & 3.44e-01 & \goodCell{} 3.43e-01 & \goodCell{} 8.49e-03 & \goodCell{} 3.31e-01 & \goodCell{} 3.44e-01 & \goodCell{} 3.19e-03 & \goodCell{} 3.36e-01 & 3,000,000 \\
  & 8 & 5.22e-01 & \goodCell{} 5.22e-01 & \goodCell{} 6.46e-03 & \goodCell{} 5.11e-01 & \goodCell{} 5.20e-01 & \goodCell{} 3.03e-03 & \goodCell{} 5.16e-01 & 3,000,000 \\
  & 11 & 6.51e-01 & \goodCell{} 6.45e-01 & \goodCell{} 1.68e-02 & \goodCell{} 6.08e-01 & \goodCell{} 6.50e-01 & \goodCell{} 2.46e-03 & \goodCell{} 6.45e-01 & 3,000,000 \\
  & 14 & 7.46e-01 & \goodCell{} 7.45e-01 & \goodCell{} 6.90e-03 & \goodCell{} 7.31e-01 & \goodCell{} 7.47e-01 & \goodCell{} 3.17e-03 & \goodCell{} 7.41e-01 & 3,000,000 \\
  & 17 & 8.15e-01 & \goodCell{} 8.08e-01 & \goodCell{} 9.79e-03 & \goodCell{} 7.88e-01 & \goodCell{} 8.15e-01 & \goodCell{} 2.20e-03 & \goodCell{} 8.11e-01 & 3,000,000 \\
  \midrule
  \multirow[c]{5}{*}{\TCP{}} & 5 & 1.90e-01 & \goodCell{} 1.90e-01 & \goodCell{} 5.38e-03 & \goodCell{} 1.78e-01 & \goodCell{} 1.90e-01 & \goodCell{} 2.13e-03 & \goodCell{} 1.86e-01 & 1,200,000 \\
  & 8 & 4.10e-01 & \goodCell{} 4.10e-01 & \goodCell{} 6.63e-03 & \goodCell{} 4.00e-01 & \goodCell{} 4.10e-01 & \goodCell{} 2.73e-03 & \goodCell{} 4.05e-01 & 1,200,000 \\
  & 11 & 5.70e-01 & \goodCell{} 5.69e-01 & \goodCell{} 5.51e-03 & \goodCell{} 5.59e-01 & \goodCell{} 5.69e-01 & \goodCell{} 2.96e-03 & \goodCell{} 5.63e-01 & 1,200,000 \\
  & 14 & 6.86e-01 & \goodCell{} 6.88e-01 & \goodCell{} 6.71e-03 & \goodCell{} 6.75e-01 & \goodCell{} 6.84e-01 & \goodCell{} 9.74e-03 & \goodCell{} 6.45e-01 & 1,200,000 \\
  & 17 & 7.71e-01 & \goodCell{} 7.68e-01 & \goodCell{} 1.22e-02 & \goodCell{} 7.40e-01 & \goodCell{} 7.71e-01 & \goodCell{} 2.52e-03 & \goodCell{} 7.66e-01 & 1,200,000 \\
  \midrule
  \FirstGridWorld{} &  & 6.18e-01 & \goodCell{} 6.17e-01 & \goodCell{} 9.57e-03 & \goodCell{} 5.98e-01 & 5.69e-01 & 1.35e-01 & 0.00e+00 & 4,000,000 \\
  \midrule
  \SecondGridWorld{} &  & 6.71e-01 & \goodCell{} 6.72e-01 & \goodCell{} 6.88e-03 & \goodCell{} 6.56e-01 & 6.83e-02 & 1.70e-01 & 0.00e+00 & 1,500,000 \\
  \midrule
  \multirow[c]{2}{*}{\SharedCoin{}} & 14 & 1.25e-01 & \goodCell{} 1.25e-01 & \goodCell{} 5.18e-03 & \goodCell{} 1.14e-01 & 1.09e-01 & 2.71e-02 & 6.15e-02 & 4,000,000 \\
  & 20 & 2.50e-01 & \goodCell{} 2.51e-01 & \goodCell{} 5.16e-03 & \goodCell{} 2.41e-01 & 2.18e-01 & 2.74e-02 & 1.55e-01 & 4,000,000 \\
  \bottomrule
 \end{tabular}
\end{table}

\ifsupresstikzfig
\else
\begin{figure}[tbp]
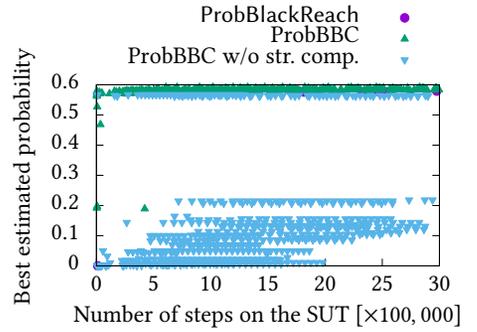
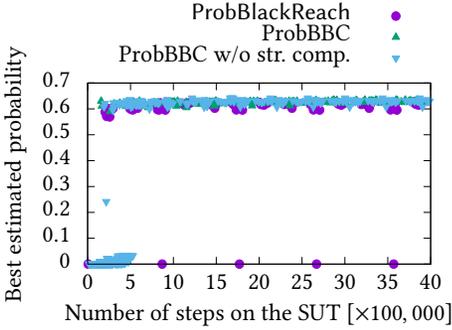
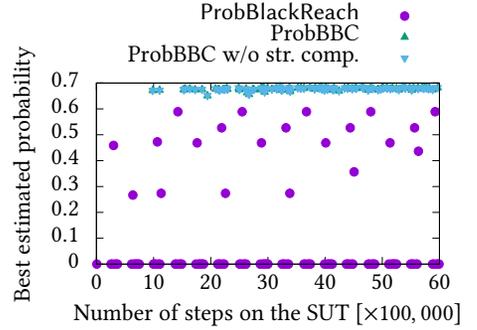
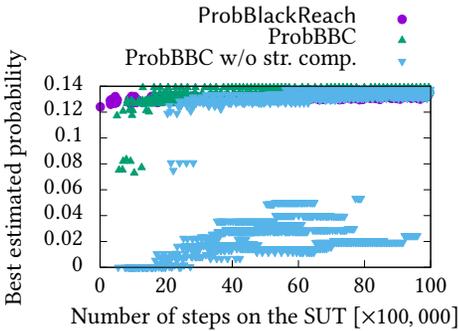
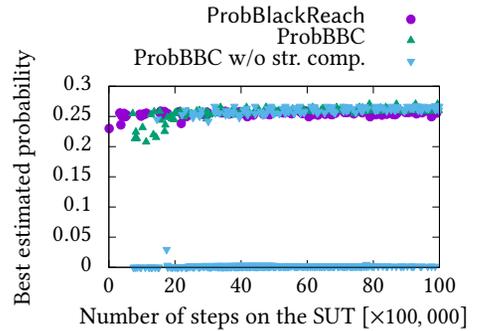

 \begin{subfigure}{.45\textwidth}
  \centering
  \scalebox{0.5}{\input{./figures/slot17.tikz}}
  \caption{\SlotMachine{} with $\Spec = \Evt_{[0, 17)}\PrTen$.}
  \label{figure:estimation_graph:slot17}
 \end{subfigure}
 \hfill
 \begin{subfigure}{.45\textwidth}
  \centering
  \scalebox{0.5}{\input{./figures/slot_limited17.tikz}}
  \caption{\SlotMachineWithSuppressedOutputs{} with $\Spec = \Evt_{[0, 17)}\PrTen$.}
  \label{figure:estimation_graph:slot_limited17}
 \end{subfigure}
 \hfill
 \begin{subfigure}{.45\textwidth}
  \centering
  \scalebox{0.5}{\input{./figures/mqtt14.tikz}}
  \caption{\MQTT{} with $\Spec = \Evt_{[0, 14)}\mathit{crash}$.}
  \label{figure:estimation_graph:mqtt14}
 \end{subfigure}
 \hfill
 \begin{subfigure}{.45\textwidth}
  \centering
  \scalebox{0.5}{\input{./figures/tcp11.tikz}}
  \caption{\TCP{} with $\Spec = \Evt_{[0, 11)}\mathit{crash}$.}
  \label{figure:estimation_graph:tcp11}
 \end{subfigure}
 \hfill
 \begin{subfigure}{.45\textwidth}
  \centering
  \scalebox{0.5}{\input{./figures/first.tikz}}
  \caption{\FirstGridWorld{} with $\Spec = \Evt_{[0, 10)}\mathit{goal}$.}
  \label{figure:estimation_graph:first}
 \end{subfigure}
 \hfill
 \begin{subfigure}{.45\textwidth}
  \centering
  \scalebox{0.5}{\input{./figures/second.tikz}}
  \caption{\SecondGridWorld{} with $\Spec = \Evt_{[0, 13)}\mathit{goal}$.}
  \label{figure:estimation_graph:second}
 \end{subfigure}
 \hfill
 \begin{subfigure}{.45\textwidth}
  \centering
  \scalebox{0.5}{\input{./figures/sc14.tikz}}
  \caption{\SharedCoin{} with $\Spec = \Evt_{[0, 14)}\mathit{finished}$.}
  \label{figure:estimation_graph:sc14}
 \end{subfigure}
 \hfill
 \begin{subfigure}{.45\textwidth}
  \centering
  \scalebox{0.5}{\input{./figures/sc20.tikz}}
  \caption{\SharedCoin{} with $\Spec = \Evt_{[0, 20)}\mathit{finished}$.}
  \label{figure:estimation_graph:sc20}
 \end{subfigure}
 \caption{Estimated probability after each execution step on the SUT.}%
 \label{figure:estimation_graph}
\end{figure}
\fi

\subsection{RQ1: Quality of the estimated probabilities}\label{section:RQ1}

In the columns ``PRISM'' and ``\ourTool{}/mean'' of \cref{table:experiment_result}, we observe that for any LTL formula in our benchmarks, the satisfaction probability estimated by \ourTool{} is close to the true probabilities on average.
We note that the estimated probabilities can be slightly larger than the true probabilities due to the statistical error in the estimation with SMC.\@

Moreover, in the columns ``PRISM'' and ``\ourTool{}/min'', we observe that for some benchmarks (\eg{} \FirstGridWorld{}, \SecondGridWorld{}, and \SharedCoin{}), 
the satisfaction probability estimated by \ourTool{} is close to the true probabilities even in the worst case.
This is likely because of the relatively small state space ($|Q| < 300$) and the small input size ($|\INPUT| \leq 4$) of the MDPs in these benchmarks,
which make the maximum size of the observation table relatively small and the \LstarMDP{} algorithm identifies the target MDP with a mild number of queries.

These results suggest that the maximum satisfaction probabilities estimated by \ourTool{} are usually close to the optimal value.
Overall, we answer RQ1 as follows.

\rqanswer{RQ1}{\ourTool{} usually estimates the maximum satisfaction probabilities close to the true one. For some benchmarks, the estimated probabilities are very close to the optimal value even in the worst case.}

\subsection{RQ2: Quality of the estimated probabilities compared with \baselineMethod{}}\label{section:experiment:comparison}

In the columns ``PRISM'' and ``\baselineMethod{}/mean'' of \cref{table:experiment_result},
we observe that \baselineMethod{} often fails to estimate the maximum satisfaction probability close to the true one.
For example, for 11 out of 24 LTL formulas, the maximum satisfaction probability of the formula estimated by \baselineMethod{}
is less than $97.5 \%$ of the true value (\ie{} not highlighted in \cref{table:experiment_result}).
This is likely because \baselineMethod{} often gets stuck on a local optimum due to the greedy nature of its trace sampling.
Such a tendency is indeed observed
in \cref{figure:estimation_graph} for some benchmarks.
In contrast, as we observe in \cref{section:RQ1}, for all the LTL formulas, the probability estimated by \ourTool{} is close to the true one.
This is because the MDP learned by the \LstarMDP{} algorithm eventually converges to the SUT, as shown in~\cite{DBLP:journals/fac/TapplerA0EL21}.
Such convergence is also observed in \cref{figure:estimation_graph}.

In the column ``\baselineMethod{}/std'' of \cref{table:experiment_result},
we observe that for \FirstGridWorld{} and \SecondGridWorld{}, the standard deviation of the probabilities estimated by \baselineMethod{} is large (\eg{} greater than $0.1$).
In the column ``\baselineMethod{}/min'' of \cref{table:experiment_result},
we also observe that for \FirstGridWorld{} and \SecondGridWorld{},
the satisfaction probability estimated by \baselineMethod{} can be 0.
This is likely because synthesizing a reasonable strategy for a grid world environment requires a relatively precise estimation of the environment,
which is not always done by \baselineMethod{}.
Such a tendency is more evident in \SecondGridWorld{}, where the environment is larger, and its precise estimation is more challenging.
In contrast, \ourTool{} always synthesizes a reasonable strategy up to our observation.
Overall, we answer RQ2 as follows.

\rqanswer{RQ2}{\baselineMethod{} often fails to estimate a reasonable maximum satisfaction probability. Moreover, it sometimes fails to learn a reasonable strategy for \FirstGridWorld{} and \SecondGridWorld{}.}

\subsection{RQ3: Robustness with respect to the observability}

In the rows ``\SlotMachine{}'' and ``\SlotMachineWithSuppressedOutputs{}'' of \cref{table:experiment_result},
we observe that the probabilities estimated by \SlotMachineWithSuppressedOutputs{} tend to be smaller with limited observations, especially when $n$ is large.
This is likely because \baselineMethod{} samples the training data with an $\varepsilon$-greedy algorithm, and their variety is limited.
When an MDP is learned, two candidate states are deemed identical if there is no clear counterexample, where outputs are used to compare the states.
Therefore, to correctly distinguish the states, there must be sufficient variety in the training data to contain an evidence.
The training data need more variety when many states have the same output, \eg{} more variety is required in \SlotMachineWithSuppressedOutputs{} than in \SlotMachine{}.

In contrast, the probabilities estimated by \ourTool{} do not have such a deviation.
This is likely because the \LstarMDP{} algorithm tries to cover various inputs to have sufficient information to estimate the transition probabilities.
Moreover, \ourTool{} tends to estimate \emph{better} probabilities with limited observability.
This is likely because of the following reason:
The MDP of \SlotMachine{} has states that have different outputs but behave the same;
For example, the states corresponding to ``bar blank apple'' and ``bar apple blank'' behave the same, but their outputs are different in \SlotMachine{};
In the MDP of \SlotMachineWithSuppressedOutputs{} these states have the same outputs, and they are deemed identical in MDP learning;
Since the target MDP under learning is virtually smaller, the \LstarMDP{} algorithm can easily converge to the optimal one;
Therefore, ProbBBC can perform better with limited observability.
We indeed observe that the resulting MDPs tend to be smaller for \SlotMachineWithSuppressedOutputs{} than \SlotMachine{} (about 50 vs. 160 states on average).
Overall, we answer RQ3 as follows.

\rqanswer{RQ3}{\ourTool{} estimates a near-optimal satisfaction probability even if the observability in the SUT is limited, whereas \baselineMethod{} often fails to estimate it in such a situation. Moreover, \ourTool{} often estimates a better probability when the observability is limited because the state space of the SUT can be virtually smaller.}

\subsection{RQ4: Efficiency of the estimation}

In \cref{figure:estimation_graph:slot17,figure:estimation_graph:slot_limited17,figure:estimation_graph:first,figure:estimation_graph:second},
we observe that the probabilities estimated by \baselineMethod{} often remain suboptimal, as discussed in \cref{section:experiment:comparison}, whereas those estimated by \ourTool{} usually converge (with some exceptions, \eg{} some executions for \MQTT{} with $n = 14$ shown in \cref{figure:estimation_graph:mqtt14}).
In contrast, in \cref{figure:estimation_graph:slot17,figure:estimation_graph:slot_limited17,figure:estimation_graph:sc14,figure:estimation_graph:sc20},
we observe that the initial rising up of the probabilities estimated by \ourTool{} is slower than that by \baselineMethod{}.
This is also likely because of the use of by an $\varepsilon$-greedy algorithm in the training data construction:
The $\varepsilon$-greedy algorithm samples the traces deemed to increase the satisfaction probability of the given LTL formula.
If such traces indeed increase the satisfaction probability, the estimated probability increases with a small number of traces.
However, this is at the cost of the quality of the final estimation, and the delay of the initial rising up is not very large except for \cref{figure:estimation_graph:slot17}.
Overall, we answer RQ4 as follows.

\rqanswer{RQ4}{\baselineMethod{} is often faster to start up than \ourTool{} at the cost of the quality of the estimated probability. Moreover, the delay of the initial rising up is usually not large.}

\subsection{RQ5: Sensitivity to the parameters}\label{section:experiment:parameter_sensitivity}

\begin{table}[tbp]
 \centering
 \caption{Mean of the estimated maximum satisfaction probabilities of $\Spec$ after the number of steps displayed in the column ``\# of steps'' of \cref{table:experiment_result}. The column $n$ shows the parameter $n$ of the LTL formulas in \cref{table:benchmarks}. We highlight the cells whose values are larger than $97.5\%$ of the true probability in \cref{table:experiment_result}.}%
 \label{table:ProbBBC_parameter_robustness}
 \scriptsize
 \begin{tabular}{llllllllll}
  \toprule
  & \multirow[c]{2}{*}{$n$} & \multicolumn{3}{c}{$N = 2500$} & \multicolumn{2}{c}{$N = 5000$} & \multicolumn{3}{c}{$N = 10000$} \\
  &  & $\Delta = 0.01$ & $\Delta = 0.025$ & $\Delta = 0.05$ & $\Delta = 0.01$ & $\Delta = 0.05$ & $\Delta = 0.01$ & $\Delta = 0.025$ & $\Delta = 0.05$ \\
  \midrule
  \multirow[c]{5}{*}{\SlotMachine{}} & 5 & \goodCell{} 8.19e-02 & \goodCell{} 8.15e-02 & \goodCell{} 8.21e-02 & \goodCell{} 8.23e-02 & \goodCell{} 8.19e-02 & \goodCell{} 8.29e-02 & \goodCell{} 8.22e-02 & \goodCell{} 8.11e-02 \\
  & 8 & \goodCell{} 3.33e-01 & \goodCell{} 3.34e-01 & \goodCell{} 3.32e-01 & \goodCell{} 3.31e-01 & \goodCell{} 3.31e-01 & \goodCell{} 3.26e-01 & \goodCell{} 3.30e-01 & \goodCell{} 3.31e-01 \\
  & 11 & \goodCell{} 4.58e-01 & \goodCell{} 4.58e-01 & \goodCell{} 4.58e-01 & \goodCell{} 4.57e-01 & \goodCell{} 4.59e-01 & \goodCell{} 4.57e-01 & \goodCell{} 4.55e-01 & \goodCell{} 4.59e-01 \\
  & 14 & \goodCell{} 4.94e-01 & \goodCell{} 4.94e-01 & \goodCell{} 4.97e-01 & \goodCell{} 4.92e-01 & \goodCell{} 4.95e-01 & \goodCell{} 4.94e-01 & \goodCell{} 4.91e-01 & \goodCell{} 4.95e-01 \\
  & 17 & \goodCell{} 5.05e-01 & \goodCell{} 5.01e-01 & \goodCell{} 5.04e-01 & \goodCell{} 5.08e-01 & \goodCell{} 5.01e-01 & \goodCell{} 5.05e-01 & \goodCell{} 5.02e-01 & \goodCell{} 5.06e-01 \\
  \midrule
  \multirow[c]{5}{*}{\SlotMachineWithSuppressedOutputs{}} & 5 & \goodCell{} 8.07e-02 & \goodCell{} 8.07e-02 & \goodCell{} 8.07e-02 & \goodCell{} 8.07e-02 & \goodCell{} 8.07e-02 & \goodCell{} 8.07e-02 & \goodCell{} 8.07e-02 & \goodCell{} 8.07e-02 \\
  & 8 & \goodCell{} 3.33e-01 & \goodCell{} 3.33e-01 & \goodCell{} 3.33e-01 & \goodCell{} 3.33e-01 & \goodCell{} 3.33e-01 & \goodCell{} 3.33e-01 & \goodCell{} 3.33e-01 & \goodCell{} 3.33e-01 \\
  & 11 & \goodCell{} 4.60e-01 & \goodCell{} 4.60e-01 & \goodCell{} 4.60e-01 & \goodCell{} 4.60e-01 & \goodCell{} 4.60e-01 & \goodCell{} 4.60e-01 & \goodCell{} 4.60e-01 & \goodCell{} 4.60e-01 \\
  & 14 & \goodCell{} 4.97e-01 & \goodCell{} 4.97e-01 & \goodCell{} 4.97e-01 & \goodCell{} 4.97e-01 & \goodCell{} 4.97e-01 & \goodCell{} 4.97e-01 & \goodCell{} 4.97e-01 & \goodCell{} 4.97e-01 \\
  & 17 & \goodCell{} 5.10e-01 & \goodCell{} 5.10e-01 & \goodCell{} 5.10e-01 & \goodCell{} 5.10e-01 & \goodCell{} 5.10e-01 & \goodCell{} 5.10e-01 & \goodCell{} 5.10e-01 & \goodCell{} 5.10e-01 \\
  \midrule
  \multirow[c]{5}{*}{\MQTT{}} & 5 & \goodCell{} 3.43e-01 & \goodCell{} 3.43e-01 & \goodCell{} 3.43e-01 & \goodCell{} 3.43e-01 & \goodCell{} 3.43e-01 & \goodCell{} 3.43e-01 & \goodCell{} 3.43e-01 & \goodCell{} 3.43e-01 \\
  & 8 & \goodCell{} 5.22e-01 & \goodCell{} 5.22e-01 & \goodCell{} 5.22e-01 & \goodCell{} 5.22e-01 & \goodCell{} 5.22e-01 & \goodCell{} 5.22e-01 & \goodCell{} 5.22e-01 & \goodCell{} 5.22e-01 \\
  & 11 & \goodCell{} 6.45e-01 & \goodCell{} 6.45e-01 & \goodCell{} 6.45e-01 & \goodCell{} 6.45e-01 & \goodCell{} 6.45e-01 & \goodCell{} 6.45e-01 & \goodCell{} 6.45e-01 & \goodCell{} 6.45e-01 \\
  & 14 & \goodCell{} 7.45e-01 & \goodCell{} 7.45e-01 & \goodCell{} 7.45e-01 & \goodCell{} 7.45e-01 & \goodCell{} 7.45e-01 & \goodCell{} 7.45e-01 & \goodCell{} 7.45e-01 & \goodCell{} 7.45e-01 \\
  & 17 & \goodCell{} 8.08e-01 & \goodCell{} 8.08e-01 & \goodCell{} 8.08e-01 & \goodCell{} 8.08e-01 & \goodCell{} 8.08e-01 & \goodCell{} 8.08e-01 & \goodCell{} 8.08e-01 & \goodCell{} 8.08e-01 \\
  \midrule
  \multirow[c]{5}{*}{\TCP{}} & 5 & \goodCell{} 1.90e-01 & \goodCell{} 1.91e-01 & \goodCell{} 1.90e-01 & \goodCell{} 1.90e-01 & \goodCell{} 1.90e-01 & \goodCell{} 1.88e-01 & \goodCell{} 1.90e-01 & \goodCell{} 1.91e-01 \\
  & 8 & \goodCell{} 4.11e-01 & \goodCell{} 4.08e-01 & \goodCell{} 4.05e-01 & \goodCell{} 4.09e-01 & \goodCell{} 4.09e-01 & \goodCell{} 4.09e-01 & \goodCell{} 4.09e-01 & \goodCell{} 4.07e-01 \\
  & 11 & \goodCell{} 5.61e-01 & \goodCell{} 5.62e-01 & \goodCell{} 5.69e-01 & \goodCell{} 5.69e-01 & \goodCell{} 5.58e-01 & \goodCell{} 5.68e-01 & \goodCell{} 5.65e-01 & \goodCell{} 5.67e-01 \\
  & 14 & \goodCell{} 6.76e-01 & \goodCell{} 6.74e-01 & \goodCell{} 6.84e-01 & \goodCell{} 6.80e-01 & \goodCell{} 6.76e-01 & \goodCell{} 6.81e-01 & \goodCell{} 6.83e-01 & \goodCell{} 6.79e-01 \\
  & 17 & \goodCell{} 7.64e-01 & \goodCell{} 7.60e-01 & \goodCell{} 7.68e-01 & \goodCell{} 7.68e-01 & \goodCell{} 7.67e-01 & \goodCell{} 7.65e-01 & \goodCell{} 7.66e-01 & \goodCell{} 7.61e-01 \\
  \midrule
  \FirstGridWorld{} &  & \goodCell{} 6.14e-01 & \goodCell{} 6.14e-01 & \goodCell{} 6.15e-01 & \goodCell{} 6.14e-01 & \goodCell{} 6.17e-01 & \goodCell{} 6.18e-01 & \goodCell{} 6.17e-01 & \goodCell{} 6.18e-01 \\
  \midrule
  \SecondGridWorld{} &  & \goodCell{} 6.73e-01 & \goodCell{} 6.71e-01 & \goodCell{} 6.68e-01 & \goodCell{} 6.70e-01 & \goodCell{} 6.70e-01 & \goodCell{} 6.69e-01 & \goodCell{} 6.70e-01 & \goodCell{} 6.69e-01 \\
  \midrule
  \multirow[c]{2}{*}{\SharedCoin{}} & 14 & \goodCell{} 1.24e-01 & \goodCell{} 1.24e-01 & \goodCell{} 1.24e-01 & \goodCell{} 1.25e-01 & \goodCell{} 1.24e-01 & \goodCell{} 1.26e-01 & \goodCell{} 1.24e-01 & \goodCell{} 1.24e-01 \\
  & 20 & \goodCell{} 2.50e-01 & \goodCell{} 2.50e-01 & \goodCell{} 2.46e-01 & \goodCell{} 2.52e-01 & \goodCell{} 2.49e-01 & \goodCell{} 2.53e-01 & \goodCell{} 2.49e-01 & \goodCell{} 2.50e-01 \\
  \bottomrule
 \end{tabular}
\end{table}

In \cref{table:ProbBBC_parameter_robustness},
we observe that for any combination of $N$ and $\Delta$, and for any LTL formula in our benchmarks, the satisfaction probability estimated by \ourTool{} is close to the true probabilities on average.
If $N$ is small or $\Delta$ is large, 
the probability of the strategy-guided comparison to return a false witness (\ie{} a trace $\trace$ whose occurrence probability is supposed to be different between $\SUT$ and $\MDP$ but not) is also large.
However, if the probability to return a false witness is not too large, the learned MDP converges in a reasonable number of steps.
Overall, we answer RQ5 as follows.

\rqanswer{RQ5}{The validation phase in \ourTool{} is not much parameter sensitive.}

\subsection{RQ6: Effect of the strategy-guided comparison}\label{section:experiment:ablation_study}

\begin{table}[tbp]
 \centering
 \caption{Summary of the maximum satisfaction probabilities of $\Spec$ estimated by \ourTool{} without the strategy-guided comparison after the number of steps displayed in the column ``\# of steps'' of \cref{table:experiment_result}. The column $n$ shows the parameter $n$ of the LTL formulas in \cref{table:benchmarks}. The columns ``mean'', ``std'', and ``min'' show the mean, the standard deviation, and the minimum value of the results. We highlight the cells if the mean of the estimated probabilities is larger than $97.5\%$ of the true probability in \cref{table:experiment_result}.}%
 \label{table:ProbBBC_vs_only-classic}
 \scriptsize
 \begin{tabular}{llrrr}
  \toprule
  & $n$ & mean & std & min \\
  \midrule
  \multirow[c]{5}{*}{\SlotMachine{}} & 5 & \goodCell{} 8.12e-02 & \goodCell{} 4.67e-03 & \goodCell{} 7.50e-02 \\
  & 8 & 2.43e-01 & 1.48e-01 & 0.00e+00 \\
  & 11 & 4.98e-02 & 1.43e-01 & 0.00e+00 \\
  & 14 & 4.85e-03 & 7.53e-03 & 0.00e+00 \\
  & 17 & 3.03e-02 & 1.12e-01 & 0.00e+00 \\
  \midrule
  \multirow[c]{5}{*}{\SlotMachineWithSuppressedOutputs{}} & 5 & \goodCell{} 8.22e-02 & \goodCell{} 4.03e-03 & \goodCell{} 7.42e-02 \\
  & 8 & \goodCell{} 3.34e-01 & \goodCell{} 8.15e-03 & \goodCell{} 3.21e-01 \\
  & 11 & 4.34e-01 & 1.02e-01 & 2.00e-04 \\
  & 14 & 4.24e-01 & 1.81e-01 & 8.00e-04 \\
  & 17 & 4.84e-01 & 1.13e-01 & 4.40e-03 \\
  \midrule
  \multirow[c]{2}{*}{\MQTT{}} & 5 & \goodCell{} 3.43e-01 & \goodCell{} 5.62e-03 & \goodCell{} 3.26e-01 \\
  & 8 & \goodCell{} 5.20e-01 & \goodCell{} 1.33e-02 & \goodCell{} 4.70e-01 \\
  \bottomrule
 \end{tabular}
 \hfill
 \begin{tabular}{llrrr}
  \toprule
  & $n$ & mean & std & min \\
  \midrule
  \multirow[c]{3}{*}{\MQTT{}} & 11 & \goodCell{} 6.45e-01 & \goodCell{} 1.42e-02 & \goodCell{} 6.16e-01 \\
  & 14 & \goodCell{} 7.43e-01 & \goodCell{} 9.48e-03 & \goodCell{} 7.16e-01 \\
  & 17 & \goodCell{} 8.05e-01 & \goodCell{} 1.52e-02 & \goodCell{} 7.82e-01 \\
  \midrule
  \multirow[c]{5}{*}{\TCP{}} & 5 & \goodCell{} 1.89e-01 & \goodCell{} 5.46e-03 & \goodCell{} 1.77e-01 \\
  & 8 & 1.66e-02 & 3.92e-02 & 0.00e+00 \\
  & 11 & 1.42e-02 & 2.61e-02 & 0.00e+00 \\
  & 14 & 1.27e-02 & 1.79e-02 & 0.00e+00 \\
  & 17 & 1.85e-02 & 2.68e-02 & 2.00e-04 \\
  \midrule
  \FirstGridWorld{} &  & \goodCell{} 6.16e-01 & \goodCell{} 1.14e-02 & \goodCell{} 5.79e-01 \\
  \midrule
  \SecondGridWorld{} &  & 1.00e-01 & 2.43e-01 & 0.00e+00 \\
  \midrule
  \multirow[c]{2}{*}{\SharedCoin{}} & 14 & 7.64e-02 & 5.88e-02 & 1.00e-03 \\
  & 20 & 1.13e-01 & 1.28e-01 & 0.00e+00 \\
  \bottomrule
 \end{tabular}
\end{table}

In the column ``mean'' of \cref{table:ProbBBC_vs_only-classic},
we observe that
\ourTool{} often fails to estimate the maximum satisfaction probability close to the true one 
if we do not use the strategy-guided comparison.
In \cref{figure:estimation_graph}, 
we observe that the initial rising up of the probabilities estimated by \ourToolOnlyClassic{} is often slower than that by \ourTool{}.
These tendencies are likely because, without the strategy-guided comparison, the validation phase in \ourTool{} is solely by uniform random sampling.
In contrast, if we use the strategy-guided comparison, the validation phase can focus on a part of the MDP relevant to maximize the satisfaction probability of the given LTL formula, which tends to accelerate the probability estimation.
We note that, unlike \baselineMethod{}, \ourTool{} remains correct even without the strategy-guided comparison thanks to the convergence of \LstarMDP{}.
Overall, we answer RQ6 as follows.

\rqanswer{RQ6}{The strategy-guided comparison in \cref{subsec:strategy_guided_comparison} usually allows to generate a controller achieving high satisfaction probability with fewer steps on the SUT.}

\subsection{RQ7: Scalability to the system complexity}

\begin{figure}[tbp] 
 \begin{subfigure}{.30\textwidth}
  \centering
  \scalebox{0.275}{\begin{tikzpicture}[gnuplot]
\tikzset{every node/.append style={font={\fontsize{25.0pt}{30.0pt}\selectfont}}}
\path (0.000,0.000) rectangle (12.500,8.750);
\gpcolor{color=gp lt color border}
\gpsetlinetype{gp lt border}
\gpsetdashtype{gp dt solid}
\gpsetlinewidth{1.00}
\draw[gp path] (3.300,2.464)--(3.480,2.464);
\node[gp node right] at (2.840,2.464) {$0$};
\draw[gp path] (3.300,3.383)--(3.480,3.383);
\node[gp node right] at (2.840,3.383) {$40$};
\draw[gp path] (3.300,4.302)--(3.480,4.302);
\node[gp node right] at (2.840,4.302) {$80$};
\draw[gp path] (3.300,5.222)--(3.480,5.222);
\node[gp node right] at (2.840,5.222) {$120$};
\draw[gp path] (3.300,6.141)--(3.480,6.141);
\node[gp node right] at (2.840,6.141) {$160$};
\draw[gp path] (3.300,7.060)--(3.480,7.060);
\node[gp node right] at (2.840,7.060) {$200$};
\draw[gp path] (3.300,7.979)--(3.480,7.979);
\node[gp node right] at (2.840,7.979) {$240$};
\draw[gp path] (3.300,2.464)--(3.300,2.644);
\node[gp node center] at (3.300,1.694) {$0$};
\draw[gp path] (4.864,2.464)--(4.864,2.644);
\node[gp node center] at (4.864,1.694) {$40$};
\draw[gp path] (6.428,2.464)--(6.428,2.644);
\node[gp node center] at (6.428,1.694) {$80$};
\draw[gp path] (7.991,2.464)--(7.991,2.644);
\node[gp node center] at (7.991,1.694) {$120$};
\draw[gp path] (9.555,2.464)--(9.555,2.644);
\node[gp node center] at (9.555,1.694) {$160$};
\draw[gp path] (11.119,2.464)--(11.119,2.644);
\node[gp node center] at (11.119,1.694) {$200$};
\draw[gp path] (3.300,7.979)--(3.300,2.464)--(11.119,2.464)--(11.119,7.979)--cycle;
\node[gp node center,rotate=-270] at (0.730,5.221) {Execution time [min.]};
\node[gp node center] at (7.209,0.539) {Number of states of the SUT $\SUT$};
\gpcolor{rgb color={0.580,0.000,0.827}}
\gpsetlinewidth{7.00}
\draw[gp path] (3.926,2.465)--(5.802,2.512)--(7.210,2.849)--(8.930,4.952)--(10.963,7.313);
\gpsetpointsize{12.00}
\gppoint{gp mark 7}{(3.926,2.465)}
\gppoint{gp mark 7}{(5.802,2.512)}
\gppoint{gp mark 7}{(7.210,2.849)}
\gppoint{gp mark 7}{(8.930,4.952)}
\gppoint{gp mark 7}{(10.963,7.313)}
\gpcolor{rgb color={0.000,0.620,0.451}}
\gpsetdashtype{gp dt 3}
\gpsetlinewidth{5.00}
\draw[gp path] (3.926,2.465)--(3.997,2.465)--(4.068,2.466)--(4.139,2.466)--(4.210,2.467)%
  --(4.281,2.468)--(4.352,2.469)--(4.423,2.470)--(4.494,2.472)--(4.565,2.474)--(4.636,2.476)%
  --(4.707,2.478)--(4.779,2.480)--(4.850,2.483)--(4.921,2.486)--(4.992,2.490)--(5.063,2.494)%
  --(5.134,2.498)--(5.205,2.503)--(5.276,2.508)--(5.347,2.513)--(5.418,2.519)--(5.489,2.526)%
  --(5.560,2.533)--(5.631,2.540)--(5.703,2.549)--(5.774,2.557)--(5.845,2.567)--(5.916,2.577)%
  --(5.987,2.587)--(6.058,2.599)--(6.129,2.611)--(6.200,2.624)--(6.271,2.637)--(6.342,2.652)%
  --(6.413,2.667)--(6.484,2.683)--(6.556,2.700)--(6.627,2.718)--(6.698,2.737)--(6.769,2.757)%
  --(6.840,2.778)--(6.911,2.799)--(6.982,2.822)--(7.053,2.846)--(7.124,2.871)--(7.195,2.897)%
  --(7.266,2.925)--(7.337,2.953)--(7.409,2.983)--(7.480,3.014)--(7.551,3.046)--(7.622,3.079)%
  --(7.693,3.114)--(7.764,3.151)--(7.835,3.188)--(7.906,3.227)--(7.977,3.268)--(8.048,3.310)%
  --(8.119,3.354)--(8.190,3.399)--(8.262,3.445)--(8.333,3.494)--(8.404,3.544)--(8.475,3.595)%
  --(8.546,3.649)--(8.617,3.704)--(8.688,3.761)--(8.759,3.820)--(8.830,3.880)--(8.901,3.943)%
  --(8.972,4.007)--(9.043,4.073)--(9.114,4.142)--(9.186,4.212)--(9.257,4.284)--(9.328,4.359)%
  --(9.399,4.436)--(9.470,4.514)--(9.541,4.595)--(9.612,4.678)--(9.683,4.764)--(9.754,4.851)%
  --(9.825,4.942)--(9.896,5.034)--(9.967,5.129)--(10.039,5.226)--(10.110,5.326)--(10.181,5.428)%
  --(10.252,5.533)--(10.323,5.640)--(10.394,5.750)--(10.465,5.863)--(10.536,5.978)--(10.607,6.096)%
  --(10.678,6.217)--(10.749,6.340)--(10.820,6.467)--(10.892,6.596)--(10.963,6.728);
\gpcolor{color=gp lt color border}
\gpsetdashtype{gp dt solid}
\gpsetlinewidth{1.00}
\draw[gp path] (3.300,7.979)--(3.300,2.464)--(11.119,2.464)--(11.119,7.979)--cycle;
\gpdefrectangularnode{gp plot 1}{\pgfpoint{3.300cm}{2.464cm}}{\pgfpoint{11.119cm}{7.979cm}}
\end{tikzpicture}
}
  \caption{linear-linear plot}%
  \label{figure:ProbBBC_scalability:linear}
 \end{subfigure}
 \hfill
 \begin{subfigure}{.30\textwidth}
  \centering
  \scalebox{0.275}{\begin{tikzpicture}[gnuplot]
\tikzset{every node/.append style={font={\fontsize{25.0pt}{30.0pt}\selectfont}}}
\path (0.000,0.000) rectangle (12.500,8.750);
\gpcolor{color=gp lt color border}
\gpsetlinetype{gp lt border}
\gpsetdashtype{gp dt solid}
\gpsetlinewidth{1.00}
\draw[gp path] (3.760,2.464)--(3.940,2.464);
\draw[gp path] (11.119,2.464)--(10.939,2.464);
\node[gp node right] at (3.300,2.464) {$0.01$};
\draw[gp path] (3.760,2.796)--(3.850,2.796);
\draw[gp path] (11.119,2.796)--(11.029,2.796);
\draw[gp path] (3.760,2.990)--(3.850,2.990);
\draw[gp path] (11.119,2.990)--(11.029,2.990);
\draw[gp path] (3.760,3.128)--(3.850,3.128);
\draw[gp path] (11.119,3.128)--(11.029,3.128);
\draw[gp path] (3.760,3.235)--(3.850,3.235);
\draw[gp path] (11.119,3.235)--(11.029,3.235);
\draw[gp path] (3.760,3.322)--(3.850,3.322);
\draw[gp path] (11.119,3.322)--(11.029,3.322);
\draw[gp path] (3.760,3.396)--(3.850,3.396);
\draw[gp path] (11.119,3.396)--(11.029,3.396);
\draw[gp path] (3.760,3.460)--(3.850,3.460);
\draw[gp path] (11.119,3.460)--(11.029,3.460);
\draw[gp path] (3.760,3.517)--(3.850,3.517);
\draw[gp path] (11.119,3.517)--(11.029,3.517);
\draw[gp path] (3.760,3.567)--(3.940,3.567);
\draw[gp path] (11.119,3.567)--(10.939,3.567);
\node[gp node right] at (3.300,3.567) {$0.1$};
\draw[gp path] (3.760,3.899)--(3.850,3.899);
\draw[gp path] (11.119,3.899)--(11.029,3.899);
\draw[gp path] (3.760,4.093)--(3.850,4.093);
\draw[gp path] (11.119,4.093)--(11.029,4.093);
\draw[gp path] (3.760,4.231)--(3.850,4.231);
\draw[gp path] (11.119,4.231)--(11.029,4.231);
\draw[gp path] (3.760,4.338)--(3.850,4.338);
\draw[gp path] (11.119,4.338)--(11.029,4.338);
\draw[gp path] (3.760,4.425)--(3.850,4.425);
\draw[gp path] (11.119,4.425)--(11.029,4.425);
\draw[gp path] (3.760,4.499)--(3.850,4.499);
\draw[gp path] (11.119,4.499)--(11.029,4.499);
\draw[gp path] (3.760,4.563)--(3.850,4.563);
\draw[gp path] (11.119,4.563)--(11.029,4.563);
\draw[gp path] (3.760,4.620)--(3.850,4.620);
\draw[gp path] (11.119,4.620)--(11.029,4.620);
\draw[gp path] (3.760,4.670)--(3.940,4.670);
\draw[gp path] (11.119,4.670)--(10.939,4.670);
\node[gp node right] at (3.300,4.670) {$1$};
\draw[gp path] (3.760,5.002)--(3.850,5.002);
\draw[gp path] (11.119,5.002)--(11.029,5.002);
\draw[gp path] (3.760,5.196)--(3.850,5.196);
\draw[gp path] (11.119,5.196)--(11.029,5.196);
\draw[gp path] (3.760,5.334)--(3.850,5.334);
\draw[gp path] (11.119,5.334)--(11.029,5.334);
\draw[gp path] (3.760,5.441)--(3.850,5.441);
\draw[gp path] (11.119,5.441)--(11.029,5.441);
\draw[gp path] (3.760,5.528)--(3.850,5.528);
\draw[gp path] (11.119,5.528)--(11.029,5.528);
\draw[gp path] (3.760,5.602)--(3.850,5.602);
\draw[gp path] (11.119,5.602)--(11.029,5.602);
\draw[gp path] (3.760,5.666)--(3.850,5.666);
\draw[gp path] (11.119,5.666)--(11.029,5.666);
\draw[gp path] (3.760,5.723)--(3.850,5.723);
\draw[gp path] (11.119,5.723)--(11.029,5.723);
\draw[gp path] (3.760,5.773)--(3.940,5.773);
\draw[gp path] (11.119,5.773)--(10.939,5.773);
\node[gp node right] at (3.300,5.773) {$10$};
\draw[gp path] (3.760,6.105)--(3.850,6.105);
\draw[gp path] (11.119,6.105)--(11.029,6.105);
\draw[gp path] (3.760,6.299)--(3.850,6.299);
\draw[gp path] (11.119,6.299)--(11.029,6.299);
\draw[gp path] (3.760,6.437)--(3.850,6.437);
\draw[gp path] (11.119,6.437)--(11.029,6.437);
\draw[gp path] (3.760,6.544)--(3.850,6.544);
\draw[gp path] (11.119,6.544)--(11.029,6.544);
\draw[gp path] (3.760,6.631)--(3.850,6.631);
\draw[gp path] (11.119,6.631)--(11.029,6.631);
\draw[gp path] (3.760,6.705)--(3.850,6.705);
\draw[gp path] (11.119,6.705)--(11.029,6.705);
\draw[gp path] (3.760,6.769)--(3.850,6.769);
\draw[gp path] (11.119,6.769)--(11.029,6.769);
\draw[gp path] (3.760,6.826)--(3.850,6.826);
\draw[gp path] (11.119,6.826)--(11.029,6.826);
\draw[gp path] (3.760,6.876)--(3.940,6.876);
\draw[gp path] (11.119,6.876)--(10.939,6.876);
\node[gp node right] at (3.300,6.876) {$100$};
\draw[gp path] (3.760,7.208)--(3.850,7.208);
\draw[gp path] (11.119,7.208)--(11.029,7.208);
\draw[gp path] (3.760,7.402)--(3.850,7.402);
\draw[gp path] (11.119,7.402)--(11.029,7.402);
\draw[gp path] (3.760,7.540)--(3.850,7.540);
\draw[gp path] (11.119,7.540)--(11.029,7.540);
\draw[gp path] (3.760,7.647)--(3.850,7.647);
\draw[gp path] (11.119,7.647)--(11.029,7.647);
\draw[gp path] (3.760,7.734)--(3.850,7.734);
\draw[gp path] (11.119,7.734)--(11.029,7.734);
\draw[gp path] (3.760,7.808)--(3.850,7.808);
\draw[gp path] (11.119,7.808)--(11.029,7.808);
\draw[gp path] (3.760,7.872)--(3.850,7.872);
\draw[gp path] (11.119,7.872)--(11.029,7.872);
\draw[gp path] (3.760,7.929)--(3.850,7.929);
\draw[gp path] (11.119,7.929)--(11.029,7.929);
\draw[gp path] (3.760,7.979)--(3.940,7.979);
\draw[gp path] (11.119,7.979)--(10.939,7.979);
\node[gp node right] at (3.300,7.979) {$1000$};
\draw[gp path] (3.760,2.464)--(3.760,2.644);
\node[gp node center] at (3.760,1.694) {$0$};
\draw[gp path] (5.232,2.464)--(5.232,2.644);
\node[gp node center] at (5.232,1.694) {$40$};
\draw[gp path] (6.704,2.464)--(6.704,2.644);
\node[gp node center] at (6.704,1.694) {$80$};
\draw[gp path] (8.175,2.464)--(8.175,2.644);
\node[gp node center] at (8.175,1.694) {$120$};
\draw[gp path] (9.647,2.464)--(9.647,2.644);
\node[gp node center] at (9.647,1.694) {$160$};
\draw[gp path] (11.119,2.464)--(11.119,2.644);
\node[gp node center] at (11.119,1.694) {$200$};
\draw[gp path] (3.760,7.979)--(3.760,2.464)--(11.119,2.464)--(11.119,7.979)--cycle;
\node[gp node center,rotate=-270] at (0.730,5.221) {Execution time [min.]};
\node[gp node center] at (7.439,0.539) {Number of states of the SUT $\SUT$};
\gpcolor{rgb color={0.580,0.000,0.827}}
\gpsetlinewidth{5.00}
\draw[gp path] (4.349,3.272)--(6.115,5.019)--(7.440,6.020)--(9.058,6.914)--(10.972,7.234);
\gpsetpointsize{12.00}
\gppoint{gp mark 7}{(4.349,3.272)}
\gppoint{gp mark 7}{(6.115,5.019)}
\gppoint{gp mark 7}{(7.440,6.020)}
\gppoint{gp mark 7}{(9.058,6.914)}
\gppoint{gp mark 7}{(10.972,7.234)}
\gpcolor{rgb color={0.000,0.620,0.451}}
\gpsetdashtype{gp dt 3}
\draw[gp path] (4.349,3.115)--(4.416,3.290)--(4.483,3.447)--(4.549,3.590)--(4.616,3.722)%
  --(4.683,3.844)--(4.750,3.957)--(4.817,4.063)--(4.884,4.162)--(4.951,4.256)--(5.018,4.345)%
  --(5.085,4.428)--(5.152,4.508)--(5.218,4.584)--(5.285,4.657)--(5.352,4.726)--(5.419,4.793)%
  --(5.486,4.857)--(5.553,4.919)--(5.620,4.978)--(5.687,5.035)--(5.754,5.090)--(5.821,5.144)%
  --(5.887,5.196)--(5.954,5.246)--(6.021,5.294)--(6.088,5.342)--(6.155,5.387)--(6.222,5.432)%
  --(6.289,5.475)--(6.356,5.518)--(6.423,5.559)--(6.490,5.599)--(6.556,5.638)--(6.623,5.677)%
  --(6.690,5.714)--(6.757,5.750)--(6.824,5.786)--(6.891,5.821)--(6.958,5.855)--(7.025,5.889)%
  --(7.092,5.922)--(7.159,5.954)--(7.225,5.986)--(7.292,6.017)--(7.359,6.047)--(7.426,6.077)%
  --(7.493,6.106)--(7.560,6.135)--(7.627,6.163)--(7.694,6.191)--(7.761,6.218)--(7.828,6.245)%
  --(7.894,6.271)--(7.961,6.297)--(8.028,6.323)--(8.095,6.348)--(8.162,6.373)--(8.229,6.397)%
  --(8.296,6.421)--(8.363,6.445)--(8.430,6.468)--(8.497,6.491)--(8.563,6.514)--(8.630,6.537)%
  --(8.697,6.559)--(8.764,6.580)--(8.831,6.602)--(8.898,6.623)--(8.965,6.644)--(9.032,6.665)%
  --(9.099,6.685)--(9.166,6.705)--(9.232,6.725)--(9.299,6.745)--(9.366,6.764)--(9.433,6.784)%
  --(9.500,6.803)--(9.567,6.821)--(9.634,6.840)--(9.701,6.858)--(9.768,6.876)--(9.835,6.894)%
  --(9.901,6.912)--(9.968,6.930)--(10.035,6.947)--(10.102,6.964)--(10.169,6.981)--(10.236,6.998)%
  --(10.303,7.015)--(10.370,7.031)--(10.437,7.047)--(10.504,7.063)--(10.570,7.079)--(10.637,7.095)%
  --(10.704,7.111)--(10.771,7.126)--(10.838,7.142)--(10.905,7.157)--(10.972,7.172);
\gpcolor{color=gp lt color border}
\gpsetdashtype{gp dt solid}
\gpsetlinewidth{1.00}
\draw[gp path] (3.760,7.979)--(3.760,2.464)--(11.119,2.464)--(11.119,7.979)--cycle;
\gpdefrectangularnode{gp plot 1}{\pgfpoint{3.760cm}{2.464cm}}{\pgfpoint{11.119cm}{7.979cm}}
\end{tikzpicture}
}
  \caption{linear-log plot}%
  \label{figure:ProbBBC_scalability:linear_log}
 \end{subfigure}
 \hfill
 \begin{subfigure}{.30\textwidth}
  \centering
  \scalebox{0.275}{\begin{tikzpicture}[gnuplot]
\tikzset{every node/.append style={font={\fontsize{25.0pt}{30.0pt}\selectfont}}}
\path (0.000,0.000) rectangle (12.500,8.750);
\gpcolor{color=gp lt color border}
\gpsetlinetype{gp lt border}
\gpsetdashtype{gp dt solid}
\gpsetlinewidth{1.00}
\draw[gp path] (3.760,2.464)--(3.940,2.464);
\draw[gp path] (11.119,2.464)--(10.939,2.464);
\node[gp node right] at (3.300,2.464) {$0.01$};
\draw[gp path] (3.760,2.796)--(3.850,2.796);
\draw[gp path] (11.119,2.796)--(11.029,2.796);
\draw[gp path] (3.760,2.990)--(3.850,2.990);
\draw[gp path] (11.119,2.990)--(11.029,2.990);
\draw[gp path] (3.760,3.128)--(3.850,3.128);
\draw[gp path] (11.119,3.128)--(11.029,3.128);
\draw[gp path] (3.760,3.235)--(3.850,3.235);
\draw[gp path] (11.119,3.235)--(11.029,3.235);
\draw[gp path] (3.760,3.322)--(3.850,3.322);
\draw[gp path] (11.119,3.322)--(11.029,3.322);
\draw[gp path] (3.760,3.396)--(3.850,3.396);
\draw[gp path] (11.119,3.396)--(11.029,3.396);
\draw[gp path] (3.760,3.460)--(3.850,3.460);
\draw[gp path] (11.119,3.460)--(11.029,3.460);
\draw[gp path] (3.760,3.517)--(3.850,3.517);
\draw[gp path] (11.119,3.517)--(11.029,3.517);
\draw[gp path] (3.760,3.567)--(3.940,3.567);
\draw[gp path] (11.119,3.567)--(10.939,3.567);
\node[gp node right] at (3.300,3.567) {$0.1$};
\draw[gp path] (3.760,3.899)--(3.850,3.899);
\draw[gp path] (11.119,3.899)--(11.029,3.899);
\draw[gp path] (3.760,4.093)--(3.850,4.093);
\draw[gp path] (11.119,4.093)--(11.029,4.093);
\draw[gp path] (3.760,4.231)--(3.850,4.231);
\draw[gp path] (11.119,4.231)--(11.029,4.231);
\draw[gp path] (3.760,4.338)--(3.850,4.338);
\draw[gp path] (11.119,4.338)--(11.029,4.338);
\draw[gp path] (3.760,4.425)--(3.850,4.425);
\draw[gp path] (11.119,4.425)--(11.029,4.425);
\draw[gp path] (3.760,4.499)--(3.850,4.499);
\draw[gp path] (11.119,4.499)--(11.029,4.499);
\draw[gp path] (3.760,4.563)--(3.850,4.563);
\draw[gp path] (11.119,4.563)--(11.029,4.563);
\draw[gp path] (3.760,4.620)--(3.850,4.620);
\draw[gp path] (11.119,4.620)--(11.029,4.620);
\draw[gp path] (3.760,4.670)--(3.940,4.670);
\draw[gp path] (11.119,4.670)--(10.939,4.670);
\node[gp node right] at (3.300,4.670) {$1$};
\draw[gp path] (3.760,5.002)--(3.850,5.002);
\draw[gp path] (11.119,5.002)--(11.029,5.002);
\draw[gp path] (3.760,5.196)--(3.850,5.196);
\draw[gp path] (11.119,5.196)--(11.029,5.196);
\draw[gp path] (3.760,5.334)--(3.850,5.334);
\draw[gp path] (11.119,5.334)--(11.029,5.334);
\draw[gp path] (3.760,5.441)--(3.850,5.441);
\draw[gp path] (11.119,5.441)--(11.029,5.441);
\draw[gp path] (3.760,5.528)--(3.850,5.528);
\draw[gp path] (11.119,5.528)--(11.029,5.528);
\draw[gp path] (3.760,5.602)--(3.850,5.602);
\draw[gp path] (11.119,5.602)--(11.029,5.602);
\draw[gp path] (3.760,5.666)--(3.850,5.666);
\draw[gp path] (11.119,5.666)--(11.029,5.666);
\draw[gp path] (3.760,5.723)--(3.850,5.723);
\draw[gp path] (11.119,5.723)--(11.029,5.723);
\draw[gp path] (3.760,5.773)--(3.940,5.773);
\draw[gp path] (11.119,5.773)--(10.939,5.773);
\node[gp node right] at (3.300,5.773) {$10$};
\draw[gp path] (3.760,6.105)--(3.850,6.105);
\draw[gp path] (11.119,6.105)--(11.029,6.105);
\draw[gp path] (3.760,6.299)--(3.850,6.299);
\draw[gp path] (11.119,6.299)--(11.029,6.299);
\draw[gp path] (3.760,6.437)--(3.850,6.437);
\draw[gp path] (11.119,6.437)--(11.029,6.437);
\draw[gp path] (3.760,6.544)--(3.850,6.544);
\draw[gp path] (11.119,6.544)--(11.029,6.544);
\draw[gp path] (3.760,6.631)--(3.850,6.631);
\draw[gp path] (11.119,6.631)--(11.029,6.631);
\draw[gp path] (3.760,6.705)--(3.850,6.705);
\draw[gp path] (11.119,6.705)--(11.029,6.705);
\draw[gp path] (3.760,6.769)--(3.850,6.769);
\draw[gp path] (11.119,6.769)--(11.029,6.769);
\draw[gp path] (3.760,6.826)--(3.850,6.826);
\draw[gp path] (11.119,6.826)--(11.029,6.826);
\draw[gp path] (3.760,6.876)--(3.940,6.876);
\draw[gp path] (11.119,6.876)--(10.939,6.876);
\node[gp node right] at (3.300,6.876) {$100$};
\draw[gp path] (3.760,7.208)--(3.850,7.208);
\draw[gp path] (11.119,7.208)--(11.029,7.208);
\draw[gp path] (3.760,7.402)--(3.850,7.402);
\draw[gp path] (11.119,7.402)--(11.029,7.402);
\draw[gp path] (3.760,7.540)--(3.850,7.540);
\draw[gp path] (11.119,7.540)--(11.029,7.540);
\draw[gp path] (3.760,7.647)--(3.850,7.647);
\draw[gp path] (11.119,7.647)--(11.029,7.647);
\draw[gp path] (3.760,7.734)--(3.850,7.734);
\draw[gp path] (11.119,7.734)--(11.029,7.734);
\draw[gp path] (3.760,7.808)--(3.850,7.808);
\draw[gp path] (11.119,7.808)--(11.029,7.808);
\draw[gp path] (3.760,7.872)--(3.850,7.872);
\draw[gp path] (11.119,7.872)--(11.029,7.872);
\draw[gp path] (3.760,7.929)--(3.850,7.929);
\draw[gp path] (11.119,7.929)--(11.029,7.929);
\draw[gp path] (3.760,7.979)--(3.940,7.979);
\draw[gp path] (11.119,7.979)--(10.939,7.979);
\node[gp node right] at (3.300,7.979) {$1000$};
\draw[gp path] (4.160,2.464)--(4.160,2.554);
\draw[gp path] (4.783,2.464)--(4.783,2.554);
\draw[gp path] (5.297,2.464)--(5.297,2.554);
\draw[gp path] (5.734,2.464)--(5.734,2.554);
\draw[gp path] (6.114,2.464)--(6.114,2.554);
\draw[gp path] (6.451,2.464)--(6.451,2.644);
\node[gp node center] at (6.451,1.694) {$40$};
\draw[gp path] (3.760,7.979)--(3.760,2.464)--(11.119,2.464)--(11.119,7.979)--cycle;
\node[gp node center,rotate=-270] at (0.730,5.221) {Execution time [min.]};
\node[gp node center] at (7.439,0.539) {Number of states of the SUT $\SUT$};
\gpcolor{rgb color={0.580,0.000,0.827}}
\gpsetlinewidth{5.00}
\draw[gp path] (3.760,3.272)--(7.832,5.019)--(9.142,6.020)--(10.213,6.914)--(11.119,7.234);
\gpsetpointsize{12.00}
\gppoint{gp mark 7}{(3.760,3.272)}
\gppoint{gp mark 7}{(7.832,5.019)}
\gppoint{gp mark 7}{(9.142,6.020)}
\gppoint{gp mark 7}{(10.213,6.914)}
\gppoint{gp mark 7}{(11.119,7.234)}
\gpcolor{rgb color={0.000,0.620,0.451}}
\gpsetdashtype{gp dt 3}
\draw[gp path] (3.760,3.115)--(3.834,3.156)--(3.909,3.197)--(3.983,3.238)--(4.057,3.279)%
  --(4.132,3.320)--(4.206,3.361)--(4.280,3.402)--(4.355,3.443)--(4.429,3.484)--(4.503,3.525)%
  --(4.578,3.566)--(4.652,3.607)--(4.726,3.648)--(4.801,3.689)--(4.875,3.730)--(4.949,3.771)%
  --(5.024,3.812)--(5.098,3.853)--(5.172,3.894)--(5.247,3.935)--(5.321,3.976)--(5.395,4.017)%
  --(5.470,4.058)--(5.544,4.099)--(5.618,4.140)--(5.693,4.181)--(5.767,4.222)--(5.841,4.263)%
  --(5.916,4.304)--(5.990,4.345)--(6.064,4.386)--(6.139,4.427)--(6.213,4.468)--(6.287,4.509)%
  --(6.362,4.550)--(6.436,4.591)--(6.510,4.632)--(6.585,4.673)--(6.659,4.714)--(6.733,4.755)%
  --(6.808,4.796)--(6.882,4.836)--(6.956,4.877)--(7.031,4.918)--(7.105,4.959)--(7.179,5.000)%
  --(7.254,5.041)--(7.328,5.082)--(7.402,5.123)--(7.477,5.164)--(7.551,5.205)--(7.625,5.246)%
  --(7.700,5.287)--(7.774,5.328)--(7.848,5.369)--(7.923,5.410)--(7.997,5.451)--(8.071,5.492)%
  --(8.146,5.533)--(8.220,5.574)--(8.294,5.615)--(8.369,5.656)--(8.443,5.697)--(8.517,5.738)%
  --(8.592,5.779)--(8.666,5.820)--(8.740,5.861)--(8.815,5.902)--(8.889,5.943)--(8.963,5.984)%
  --(9.038,6.025)--(9.112,6.066)--(9.186,6.107)--(9.261,6.148)--(9.335,6.189)--(9.409,6.230)%
  --(9.484,6.271)--(9.558,6.312)--(9.632,6.353)--(9.707,6.394)--(9.781,6.435)--(9.855,6.476)%
  --(9.930,6.517)--(10.004,6.558)--(10.078,6.599)--(10.153,6.639)--(10.227,6.680)--(10.301,6.721)%
  --(10.376,6.762)--(10.450,6.803)--(10.524,6.844)--(10.599,6.885)--(10.673,6.926)--(10.747,6.967)%
  --(10.822,7.008)--(10.896,7.049)--(10.970,7.090)--(11.045,7.131)--(11.119,7.172);
\gpcolor{color=gp lt color border}
\gpsetdashtype{gp dt solid}
\gpsetlinewidth{1.00}
\draw[gp path] (3.760,7.979)--(3.760,2.464)--(11.119,2.464)--(11.119,7.979)--cycle;
\gpdefrectangularnode{gp plot 1}{\pgfpoint{3.760cm}{2.464cm}}{\pgfpoint{11.119cm}{7.979cm}}
\end{tikzpicture}
}
  \caption{log-log plot}%
  \label{figure:ProbBBC_scalability:log}
 \end{subfigure}
 \caption{State size of the SUT and the mean execution time to estimate a probability larger than $97.5\%$ of the true one, per SUT size in \RandomGridWorld{}. The dashed curve is $y = 1.99 \times 10^{-4} x^{3.38}$, obtained by regression.}%
 \label{figure:ProbBBC_scalability}
\end{figure}
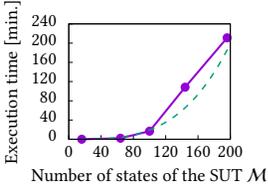
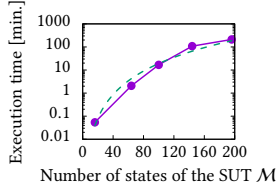
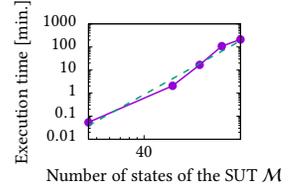
In \cref{figure:ProbBBC_scalability}, we observe that, on average, \ourTool{} can find a near-optimal controller for SUTs with 196 states within around $3.5$ hours.
In \cref{figure:ProbBBC_scalability}, we also observe that the time to find a near-optimal controller nearly follows a polynomial curve, which coincides with the polynomial complexity of \Lstar{}~\cite{DBLP:journals/iandc/Angluin87}.
Therefore, we answer RQ7 as follows.

\rqanswer{RQ7}{Experiments suggest that \ourTool{}  can generate near-optimal controllers in polynomial time with respect to system size.}

In the execution log, we find that ProbBBC often takes quite a long time to make the observation table satisfy the requirements to construct an MDP.\@
These requirements prevent us from generating MDPs different from the target one, for example, if the transition function is inconsistent.
Nevertheless, It is a possible future direction to construct an MDP even if the observation table does not satisfy some of the requirements, \eg{} using passive MDP learning~\cite{DBLP:journals/corr/abs-1212-3873,DBLP:journals/ml/MaoCJNLN16}.
Such an MDP is imprecise but may be still useful for generating a near-optimal strategy.

\section{Conclusions and future work}\label{sec:conclusion}

We propose \emph{probabilistic black-box checking (ProbBBC)}, an extension of BBC for stochastic systems\LongVersion{, by combining
active MDP learning, probabilistic model checking, and statistical hypothesis testing}.
We conducted experiments to evaluate the performance of ProbBBC using benchmarks related to CPS or IoT scenarios.
Our experiment results suggest that ProbBBC outperforms an existing approach~\cite{DBLP:journals/fmsd/AichernigT19}, especially when the SUT is complex (\eg{} \SecondGridWorld{}) or the observability is limited (\eg{} \SlotMachineWithSuppressedOutputs{}).

As formulas of specifications, ProbBBC uses safety LTL.\@
Using other logic, such as probabilistic computation tree logic (PCTL)~\cite{DBLP:reference/mc/BaierAFK18}, or PCTL*~\cite{DBLP:conf/fsttcs/BiancoA95}, is one of the future directions.
Another future direction is conducting a case study with larger and more practical benchmarks.

%
%
\begin{FinalVersionBlock}
\begin{acks}
This work was partially supported by
JST CREST Grant No.\ JPMJCR2012,
JST PRESTO Grant No.\ JPMJPR22CA,
JST ACT-X Grant No.\ JPMJAX200U, and
JSPS KAKENHI Grant No.\ 22K17873 \& 19H04084.
\end{acks}
\end{FinalVersionBlock}

\bibliographystyle{ACM-Reference-Format}
\ifdefined\VersionLong%
	\newcommand{\IJFCS}{International Journal of Foundations of Computer Science}
	\newcommand{\JLAP}{Journal of Logic and Algebraic Programming}
	\newcommand{\LNCS}{Lecture Notes in Computer Science}
	\newcommand{\STTT}{International Journal on Software Tools for Technology Transfer}
	\newcommand{\ToPNoC}{Transactions on Petri Nets and Other Models of Concurrency}
\else
	\newcommand{\IJFCS}{International Journal of Foundations of Computer Science}
	\newcommand{\JLAP}{Journal of Logic and Algebraic Programming}
	\newcommand{\LNCS}{LNCS}
	\newcommand{\STTT}{International Journal on Software Tools for Technology Transfer}
	\newcommand{\ToPNoC}{Transactions on Petri Nets and Other Models of Concurrency}
\fi
\bibliography{refs}


\ifdefined\VersionWithComments%
\newpage
\listoftodos{}
\fi

\ifdefined\WithReply%
	\clearpage
	\newpage
	\input{letter2.tex}
\fi

\end{document}